\documentclass[review,sort&compress]{elsarticle}

\usepackage{lipsum}
\makeatletter
\def\ps@pprintTitle{%
	\let\@oddhead\@empty
	\let\@evenhead\@empty
	\def\@oddfoot{}%
	\let\@evenfoot\@oddfoot}
\makeatother

\usepackage{lineno}
\modulolinenumbers[5]

\journal{arXiv.org}   
\usepackage{caption}
\usepackage{graphicx}
\usepackage{epstopdf}
\usepackage{amsmath} 
\usepackage{amssymb}  
\usepackage{amsfonts}
\usepackage{amsthm}
\usepackage{mathtools}
\usepackage{stfloats}
\usepackage{booktabs}

\newtheorem{theorem}{Theorem}
\theoremstyle{remark}
\newtheorem{remark}{Remark}

\begin{document}

\begin{frontmatter}

\title{Energy Management of a Hybrid Photovoltaic-Fuel Cell Grid-Connected Power Plant with Low Voltage Ride-Through Control}


\author[mymainaddress]{Adeel Sabir\corref{mycorrespondingauthor}}
\cortext[mycorrespondingauthor]{Corresponding author}
\ead{adeelsabir@uohb.edu.sa}

\address[mymainaddress]{P.O. Box 1803, University of Hafr Al Batin, Hafr Al Batin, Saudi Arabia}

\begin{abstract}
An energy management scheme is presented for a grid-connected hybrid power system comprising of a photovoltaic generator as the primary power source and fuel-cell stacks as backup generation. Power production is managed between the two sources such that a flexible operation is achieved, allowing the hybrid power system to supply a desired power demand by the grid operator. In addition, the energy management algorithm and the control system are designed such that the hybrid power system supports the grid in case of both symmetrical and asymmetrical voltage sags, thus, adding low voltage ride-through capability; a requirement imposed by a number of modern grid codes on distributed generation. During asymmetrical voltage sags, the injected active power is kept constant and grid currents are maintained sinusoidal with low harmonic content without requiring a phase locked loop or positive-negative sequence extraction, hence, lowering the computational complexity and design requirements of the control system. Several test case scenarios are simulated using detailed component models using the SimPowerSystems\textsuperscript{TM} toolbox of MATLAB/Simulink computing environment to demonstrate effectiveness of the proposed energy management control system under normal operating conditions and voltage sags. 
\end{abstract}

\begin{keyword}
Fuel cells\sep photovoltaics\sep energy management\sep low voltage ride-through\sep robust control  
\end{keyword}

\end{frontmatter}


\section{Introduction}
The intermittent nature of power generated from photovoltaic (PV) generators has led to the emergence of numerous hybrid power system solutions, where the PV resource is combined with another, more reliable power source like diesel generators, battery energy storage (BES) or fuel cell (FC) generators to mitigate the fluctuating PV power output due to weather and climatic conditions. Combining PV generation with FCs is one of the more attractive options due to a number of desirable FC features such as modularity, high power output and efficiency, flexibility, cogeneration, low maintenance cost and quite operation. Most importantly, FCs are a clean source of energy and in addition to serving as backup generators, they can also be used as primary sources of energy and have important implications toward the development of the smart grid concept \cite{Rahman1988,Bernstein2013}. 

In recent literature, a variety of hybrid PV-FC power system solutions have been proposed, where energy output is managed between the PV and FC generators to smoothen the power output and follow a desired load profile. The feasibility of this concept has been studied and established by some of the earlier works (see \cite{Rahman1988} and the references therein). The literature on hybrid FC power systems under review can be broadly classified into stand-alone \cite{El-Shatter2002,Jiang2006,Uzunoglu2009,Hatti2011,Saravanan2014,Kamal2016} and grid-connected \cite{Hajizadeh2010,Lajnef2013,Mohammadi2013,Bayrak2014,Eid2014,Patra2015,Abadlia2017} solutions. Our research focuses on a grid-connected solution, hence, some salient features of the relevant existing techniques are outlined below.

In \cite{Hajizadeh2010}, the authors have presented the control of a hybrid FC-supercapacitor (SC) distributed generation (DG) system during unbalanced voltage sags, using a  fuzzy logic based energy management scheme (EMS) and controllers. In \cite{Lajnef2013} and \cite{Bayrak2014}, the authors focus on the modeling of a hybrid PV-FC power system for grid-connected applications. An EMS for a grid-connected PV-FC-BES microgrid is proposed in \cite{Mohammadi2013} where the authors have applied fuzzy sliding mode (SM) control for the inner level current controllers, injecting controlled active and reactive powers for supporting the grid under normal operating conditions. In \cite{Eid2014}, an EMS for a PV-Wind-FC microgrid is proposed that coordinates the power output between the three sources, and a combination of the conventional proportional plus integral (PI) and hysteresis control are applied. In \cite{Patra2015}, an EMS for a PV-FC-BES system with backstepping maximum power point tracking (MPPT) control for the PV power source is proposed. In \cite{Abadlia2017}, the authors have contributed a SM control strategy for a hybrid PV-FC power system.

An increased focus toward the developing smart grid paradigm requires more flexibility and reliability from DG resources, but at the same time their increased presence in the modern grid poses an additional set of challenges to the stability and reliability of the grid. As a result, grid codes are increasingly mandating that the DG resources stay connected to the grid in case of abnormal conditions like voltage sags. This is to ensure that large amounts of generation is not lost during contingencies; a situation that can severely impact grid's stability. While the techniques proposed in the aforementioned literature are quite interesting and efficient, they deal largely with operation of the hybrid power systems under normal conditions; the requirement of grid support during abnormalities (e.g., asymmetrical voltage sags) has not been explored in those works and a majority of EMS techniques for hybrid power system reported in literature. An exception is \cite{Hajizadeh2010} where operation under asymmetrical voltage sags has been considered, but the authors have not provided sufficient details about power injection into the grid. Their technique also relies on positive-negative (PN) sequence calculation for voltages and currents, that requires a carefully designed phase locked loop (PLL) and imposes additional computational complexity and design requirements on the control system. Additionally, some of the existing solutions \cite{Hajizadeh2010,Eid2014,Abadlia2017} do not adequately describe the parameter selection details for the control system.

In light of this discussion, the aim of our research is to propose an efficient EMS for a grid-connected hybrid PV-FC power source that can effectively coordinate the power output of the PV and FC generators under a variety of scenarios. Moreover, the control system must be able to ride through voltage sags and provide dynamic grid support, in line with some modern grid code requirements \cite{Troester2009}. We seek a control system that can be systematically synthesized and have low computational complexity and design requirements. To this end, we propose an EMS with a robust, low voltage ride-through (LVRT) enabled control strategy for a grid-connected hybrid PV-FC power system with PV generation as the primary power source and FCs acting as backup generation, providing the energy deficit. The key contributions of this work are as follows:
\begin{itemize}
	\item An efficient EMS is proposed that can follow a desired load profile by effectively coordinating power between PV and FC generators under both normal conditions and voltage sags. The EMS along with the its control system is able to provide dynamic support to the grid in case of symmetrical and asymmetrical voltage sags.
	\item Controller gains are systematically synthesized by solving linear matrix inequalities (LMIs) that can be solved easily using any commercially available solver. Performance and convergence rate constraints are incorporated in the design to enhance its transient response and dynamic behavior. Parameter selection guidelines are also given.
	\item During asymmetrical voltage sags, real power delivered to the grid is kept constant and the injected currents are kept sinusoidal, without requiring a PLL or the PN sequence components of unbalanced voltages or currents, that are often employed for constant real power injection \cite{Mirhosseini2016}. This leads to a significant reduction in computational complexity and design requirements of the control strategy.
	\item Uncertainties, unmodeled dynamics and disturbances are formally treated to robustify the control system. Moreover, controller design does not rely on the dynamics of PV or FC generators, thus, adding to the robustness and disturbance immunity of the control scheme.
\end{itemize} 
The rest of this paper is organized as follows: description and model of the benchmark hybrid PV-FC system are given in Section \ref{sec:model}. EMS design and operation are detailed in Section \ref{sec:ems}. Control system design and details are described in Section \ref{sec:controldesign}. Simulation results are reported in Section \ref{sec:simulations}. Lastly, the paper is concluded in Section \ref{sec:conclusion}.
 
\section{Grid-Connected PV-FC System Model}\label{sec:model}
\subsection{Hybrid System Description}
Figure \ref{fig:pvfcsld} depicts the simplified schematic of the hybrid PV-FC system considered in this work. It consists of a $100~\mathrm{kW}$ PV array: $66$ parallel $5$-module strings, each module having a rating of $305~\mathrm{W}$ connected through a dc-dc boost converter to the dc-ac voltage source converter (VSC). Two $50~\mathrm{kW}$ solid oxide FC stacks with their dedicated dc-dc boost converters are connected in parallel to the PV array's dc-dc boost converter at the dc-link with a capacitor having a nominal value $C=12~\mathrm{mF}$. The VSC with $R_{on}=1~\mathrm{m\Omega}$, is connected through a three phase filter with nominal values $R_{n}=1~\mathrm{m\Omega}$ and $L_{n}=0.25~\mathrm{mH}$, and a $220~\mathrm{kVA}$, $260\mathrm{V}/25\mathrm{kV}$ transformer having resistance $R_1=R_2=0.001~\mathrm{pu}$ and inductance $L_1=L_2=0.03~\mathrm{pu}$ that connects to the $60~\mathrm{Hz}$ grid. A switching frequency of $6~\mathrm{kHz}$ is used for the pulsewidth modulation (PWM) signals. A dump load is connected at the output of the VSC for cases when power generation exceeds demand. The PV array along with its boost converter make up the PV generator while the FCs with their converters make up the FC generator. 
\begin{figure}[t!]
	\centering
	\includegraphics[scale=0.5]{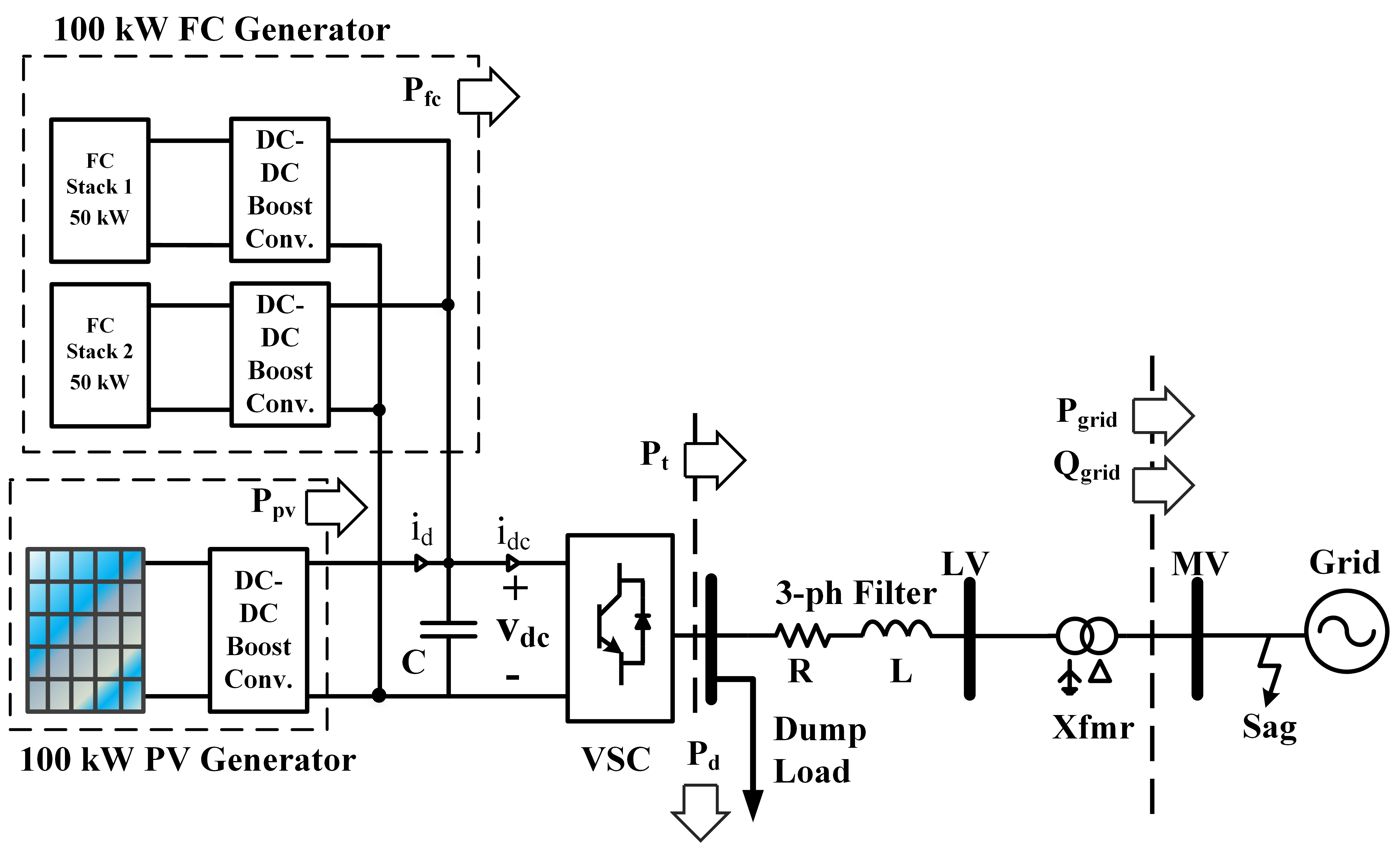}
	\caption{Grid-connected hybrid PV-FC system.}
	\label{fig:pvfcsld}
\end{figure}
Also shown in Fig. \ref{fig:pvfcsld} are the powers exchanged at different stages of the system. $P_{pv}(t)$ is the power output of the PV generator while $P_{fc}(t)$ is that of the FC generator. $P_{t}(t)$ is the power output of the VSC, $P_{d}(t)$ is the power consumed by the dump load, and $P_{grid}(t)$ and $Q_{grid}(t)$ are the real and reactive powers injected into the grid, respectively. LV denotes the low voltage bus while MV denotes the medium voltage bus; it is also where the voltage sags are applied. 

\subsection{Model}
PV cells are modeled using a single diode and including the series and shunt resistance parasitic effects, and the model reported in \cite{Zhu2002} is used for solid oxide FCs. DC-link voltage and grid current dynamics can be modeled in the stationary $\alpha\beta$ reference frame as:
\begin{equation}	\label{eq:model2}
\begin{aligned}
(C+\Delta C)\frac{d}{dt}v_{dc}(t)&=i_{d}(t)-i_{dc}(t)+ \zeta_{dc}(t),\\
\frac{d}{dt}i_\alpha(t)&=-\frac{R}{L}i_\alpha(t)+\frac{1}{L}u_{t\alpha}(t)-\frac{1}{L}V_{grid,\alpha}(t)+\zeta_{\alpha}(t),\\
\frac{d}{dt}i_\beta(t)&=-\frac{R}{L}i_\beta(t)+\frac{1}{L}u_{t\beta}(t)-\frac{1}{L}V_{grid,\beta}(t)+\zeta_{\beta}(t)
\end{aligned}	
\end{equation}
where $v_{dc}(t)$ is the dc voltage, $i_{d}(t)$ is the current output from the boost converter, $i_{dc}(t)$ is the current going into the VSC, $\Delta C$ is the additive uncertainty in the capacitance, and $\zeta_{dc}(t)$ represents all the disturbances and unmodeled quantities affecting $v_{dc}(t)$ dynamics. In the second and third equations of \eqref{eq:model2}, $i_j(t)$ are the grid current, $u_{tj}(t)$ are the VSC terminal voltages, $V_{grid,j}(t)$ are the grid voltages, $\zeta_j(t)$ denote the disturbances and unmodeled effects in grid current dynamics, and $j=\alpha,\beta$ denote the $\alpha\beta$ frame of reference. $R$ and $L$ are considered uncertain and their uncertainties are characterized in a later section. Resistance and inductance of the filter and transformer are lumped together into $R$ and $L$ in \eqref{eq:model2}, respectively, for ease of control design. All uncertainties and disturbances are assumed to be bounded. 

\section{Energy Management Scheme}	\label{sec:ems}
The EMS acts as a supervisory mechanism that manages the power generation of the hybrid PV-FC system by coordinating between the PV and FC generators. It also manages the system's power assignment during voltage sags by working in tandem with the sag detection mechanism. PV generation is the primary source of power while the FC acts as secondary generation supplying only the deficit power.

Under normal operation, the EMS reads the real and reactive power demands $P^{*}(t)$ and $Q^{*}(t)$ from the grid operator (GO). It then checks to see whether the real power demand $P^{*}(t)$ is less than or equal to the maximum power of the hybrid PV-FC system. If $P^{*}(t)$ exceeds the available maximum power, then it is set to the peak system power i.e., the available PV power $P_{pv}(t)$ plus the rated fuel cell power $P_{fc}^{rated}=100~\mathrm{kW}$. This check is included to supply the bulk of load demand from clean power sources i.e., PV and FC generators. In case a conventional dispatchable generator is included in the system e.g., a diesel generator, the unmet deficit can be assigned to it only after maximizing the power output of clean energy sources, resulting in lesser emissions. It is assumed that if a power deficit exists after extracting maximum power from the hybrid system, the GO is able to meet it through dispatchable generation or alternate sources. 

Following this step, the EMS checks if the apparent power reference $S^{*}(t)$ is within its limit $S^{max}=220~\mathrm{kVA}$ to ensure compatibility with the transformer and VSC power ratings. If the limit is exceeded, $Q^{*}(t)$ is adjusted, otherwise $Q^*(t)$ remains unchanged. At this stage, the EMS sets the FC generator's power reference $P_{fc}^{*}(t)$ and the dump load's reference power consumption $P_d^{*}(t)$. In case the available PV power $P_{pv}(t)$ exceeds demand $P^{*}(t)$, the surplus power is supplied to the dump load by setting it equal to $P_{d}^{*}(t)$ and $P_{fc}^{*}(t)$ is set to zero. The grid's reactive power reference $Q_{grid}^{*}(t)$ is then set equal to $Q^{*}(t)$. The EMS is designed to prioritize real power over reactive power delivery during normal conditions. 

If a voltage sag is detected, EMS sets the references for both real and reactive grid powers $P_{grid}^{*}(t)$ and $Q_{grid}^{*}(t)$, respectively. PV power production is curtailed by controlling its boost converter, described in a later section. The method used in \cite{Merabet2017} is employed for voltage sag detection and calculation of power references $P_{grid}^{*}(t)$ and $Q_{grid}^{*}(t)$ during sags. Apparent power limit is dynamically calculated to ensure that currents do not exceed their rated values during voltage sags. Flowchart of the EMS is depicted in Fig. \ref{fig:flowchart}.
\begin{figure}[t!]
	\centering
	\includegraphics[scale=0.7]{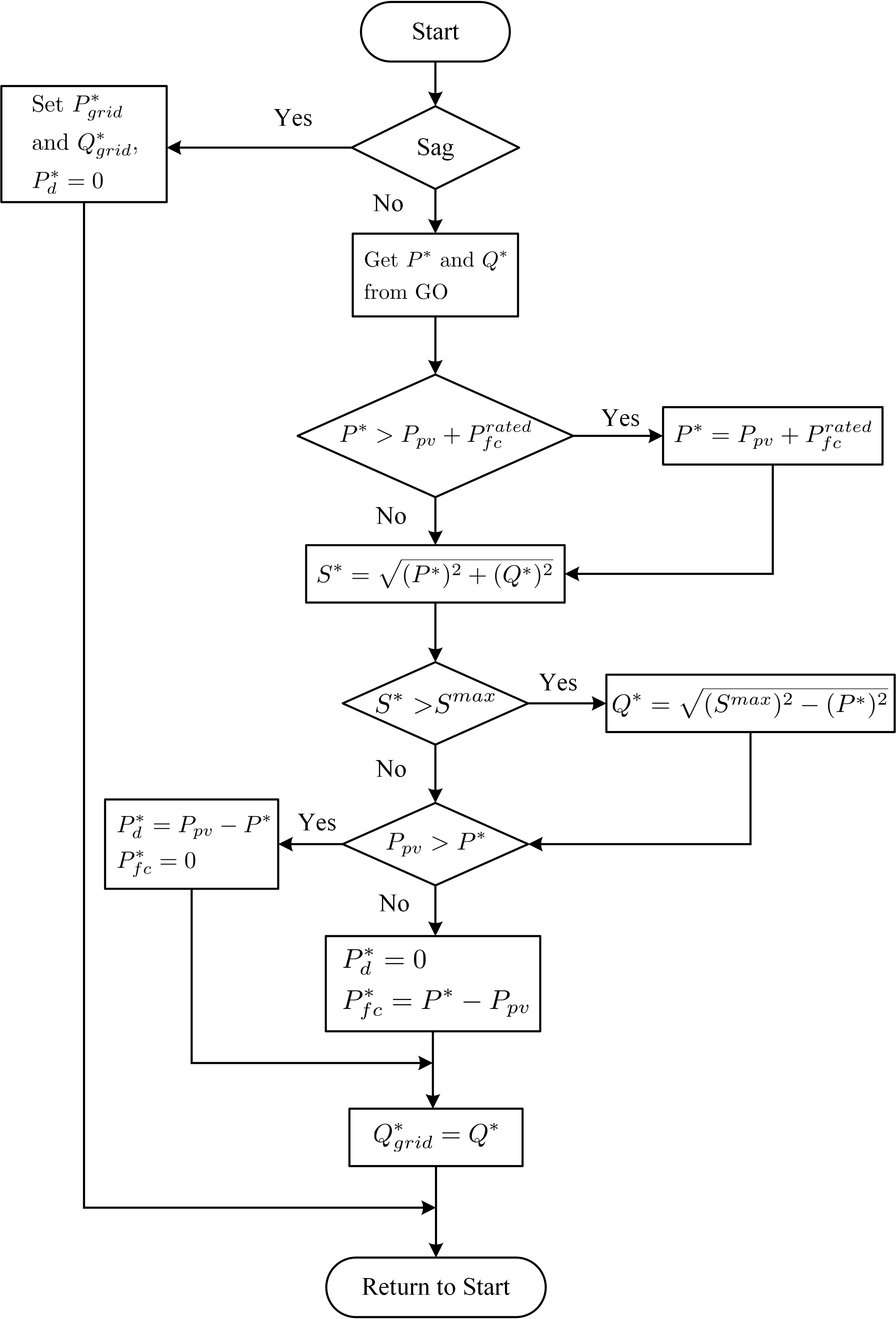}
	\caption{EMS flowchart.}
	\label{fig:flowchart}
\end{figure}

\subsection{Dump Load}
The dump load is used to maintain the active power balance in the system, in case the generation exceeds load demand. In practical systems, the dump load can be a space or water-heating system, or an electrolyzer that produces hydrogen fuel for the FCs \cite{Mendis2014,Abadlia2017}. 

\section{Control System Design and Details}\label{sec:controldesign}
We propose a combination of disturbance rejection control (DRC) and repetitive control for the hybrid PV-FC system due to their simplicity, robustness properties and ease of design. DRC is based on actively estimating the disturbances through an observer and cancelling them \cite{Han2009}. Repetitive control \cite{Inoue1981} is based on the internal model principal \cite{Francis1975} and is suitable for tracking arbitrary periodic signals. We show that a combination of DRC and repetitive control leads to a simple yet highly robust and efficient control system. The dc-link voltage $v_{dc}(t)$ is regulated using DRC while the grid currents are tracked using repetitive control. Details of controller design are presented in the following subsections.

\subsection{DC Voltage Control} \label{ssec:vdccontrol}
Consider the following model for $v_{dc}^{2}(t)$ dynamics \cite{Yazdani2010}:
\begin{equation}
\begin{aligned}
\frac{1}{2}(C+\Delta C)\frac{dv_{dc}^2(t)}{dt}&=P_{pv}(t)-P_{grid}(t)\\
&+P_{loss}(t)+\frac{2LP_{grid}(t)}{3\hat{V}_{grid}^2}\frac{dP_{grid}(t)}{dt}\\
&+\frac{2LQ_{grid}(t)}{3\hat{V}_{grid}^2}\frac{dQ_{grid}(t)}{dt}+\bar{\chi}_{dc}(t)
\end{aligned}
\label{eq:vdc2model1}
\end{equation}
where $\hat{V}_{grid}$ is the grid voltage amplitude, $P_{loss}(t)$ is the power lost due to switching, conduction and line resistances, and $\bar{\chi}_{dc}(t)$ represents the unmodeled effects and disturbances in $v_{dc}^2(t)$ dynamics. Equation \eqref{eq:vdc2model1} can be rewritten as 
\begin{equation}
\frac{dv_{dc}^2(t)}{dt}=\frac{2}{C}(P_{pv}(t)-P_{grid}(t))+\bar{\chi}_{1dc}(t)
\label{eq:vdc2model2}
\end{equation}  
where the term 
\begin{equation}
\begin{aligned}
\bar{\chi}_{1dc}(t)&=\frac{2}{C}\Big(P_{loss}(t)+\frac{2LP_{grid}(t)}{3\hat{V}_{grid}^2}\frac{dP_{grid}(t)}{dt}\\
&+\frac{2LQ_{grid}(t)}{3\hat{V}_{grid}^2}\frac{dQ_{grid}(t)}{dt}+\bar{\chi}_{dc}(t)-\Delta C\dot{v}_{dc}^{2}(t)/2\Big)
\end{aligned}
\label{eq:vdc2dist1}
\end{equation}  
is treated as a disturbance. Voltage $v_{dc}(t)$ can be controlled by regulating $P_{grid}(t)$ at its reference $P_{grid}^*(t)$ \cite{Yazdani2010}. Assuming fast power tracking, we have $P_{grid}(t)\simeq P_{grid}^{*}(t)$. From \eqref{eq:vdc2model2}, dynamics of $v_{dc}(t)$ can be written as
\begin{equation}
\frac{dv_{dc}(t)}{dt}=\frac{1}{Cv_{dc}(t)}(P_{pv}(t)-P_{grid}^{*}(t))+\chi_{dc}(t)
\label{eq:vdcmodel1}
\end{equation}
where $\chi_{dc}(t)=\bar{\chi}_{1dc}(t)/(2v_{dc}(t))$. Assuming that $P_{pv}(t)$ is measured and setting
\begin{equation}
P_{grid}^{*}(t)=P_{pv}(t)+Cv_{dc}(t)u(t)
\label{eq:Pgrid1}
\end{equation}
with $u(t)$ as an auxiliary control variable yields
\begin{equation}
\frac{dv_{dc}(t)}{dt}=-u(t)+\chi_{dc}(t).
\label{eq:vdcmodel2}
\end{equation}
Now, the first equation in \eqref{eq:model2} is rewritten as
\begin{equation}
\frac{dv_{dc}(t)}{dt}=\frac{1}{C}(i_{d}(t)-i_{dc}^{*}(t))+ \bar{\xi}_{dc}(t)
\label{eq:vdc3}
\end{equation}
where $\bar{\xi}_{dc}(t)=(1/C)(\zeta_{dc}(t)-\Delta C\dot{v}_{dc}(t))$ and $i_{dc}(t)$ has been replaced by $i_{dc}^{*}(t)$, a control variable. Let 
\begin{equation}
i_{dc}^{*}(t)=Cu(t)
\label{eq:idc1}
\end{equation}
and $i_d(t)$ be an unknown disturbance.Then, \eqref{eq:vdc3} becomes
\begin{equation}
\frac{dv_{dc}(t)}{dt}=-u(t)+\xi_{dc}(t)
\label{eq:vdcmodel3}
\end{equation}
where  $u(t)$ is a auxiliary control variable and
\begin{equation}
\xi_{dc}(t)=\frac{1}{C}i_d(t) + \bar{\xi}_{dc}(t)
\label{eq:xidc1}
\end{equation}
is the new unknown disturbance term. From the perspective of disturbance rejection control, equations \eqref{eq:vdcmodel2} and \eqref{eq:vdcmodel3} are identical and either one can be used for controller design. From \eqref{eq:vdcmodel3}, suppose 
\begin{equation}
\begin{aligned}
\dot{v}_{dc}(t)&=-u(t)+\xi_{dc}(t),\\
\dot{\xi}_{dc}(t)&=g(t),
\end{aligned}
\label{eq:vdcmodel4}
\end{equation}
where $g(t)$ is an unknown auxiliary signal to be determined. We can write \eqref{eq:vdcmodel4} in matrix form as
\begin{equation}
\begin{aligned}
\dot{x}_{dc}(t)&=A_{dc}x_{dc}(t)+B_{dc}u(t)+B_{\xi}g(t),\\
y_{dc}(t)&=C_{dc}x_{dc}(t),
\end{aligned}
\label{eq:vdcmodel5}
\end{equation}
with 
\begin{equation}
\begin{aligned}
x_{dc}(t)&=\begin{bmatrix}v_{dc}(t)\\ \xi_{dc}(t)\end{bmatrix},~
A_{dc}=\begin{bmatrix}0&1\\0&0\end{bmatrix},~
B_{dc}=\begin{bmatrix}-1\\0\end{bmatrix},\\
B_{\xi}&=\begin{bmatrix}0\\1\end{bmatrix},
~C_{dc}=\begin{bmatrix}1&0\end{bmatrix}.
\end{aligned}
\label{eq:xdc1}
\end{equation}
We design an observer with dynamics
\begin{equation}
\dot{\hat{x}}_{dc}(t)=A_{dc}\hat{x}_{dc}(t)+B_{dc}u(t)+K_{dc}^{-1}L_{dc}C_{dc}(x_{dc}(t)-\hat{x}_{dc}(t)),
\label{eq:vdcobs1}
\end{equation}
where $\hat{x}_{dc}(t)=\begin{bmatrix}\hat{v}_{dc}(t)&\hat{\xi}_{dc}(t)\end{bmatrix}^T$ is the observer's state vector, $K_{dc}\in\mathbb{R}^{2\times  2}$ is a symmetric positive definite matrix, $L_{dc}\in\mathbb{R}^{2\times 1}$ is an arbitrary real matrix, and $\hat{v}_{dc}(t)$ and $\hat{\xi}_{dc}(t)$ are the estimates of $v_{dc}(t)$ and $\xi_{dc}(t)$, respectively. The dynamics of observation error $e_{obs}(t)=x_{dc}(t)-\hat{x}_{dc}(t)$ can be obtained by subtracting \eqref{eq:vdcobs1} from \eqref{eq:vdcmodel5} as
\begin{equation}
\dot{e}_{obs}(t)=(A_{dc}-K_{dc}^{-1}L_{dc}C_{dc})e_{obs}(t) + B_{\xi}g(t).
\label{eq:eo1}
\end{equation}
Our objective is to find matrix variables $K_{dc}$ and $L_{dc}$ so that the when $g(t)=0$, dynamics \eqref{eq:eo1} are exponentially stable, and when $g(t)\neq 0$, the $\mathcal{L}_2$ gain
\begin{equation} \label{eq:L2_gain}
\frac{\begin{Vmatrix}C_{dc}\hat{x}_{dc}(t)-y_{dc}(t)\end{Vmatrix}_{\mathcal{L}_2}}{\begin{Vmatrix}g(t)\end{Vmatrix}_{\mathcal{L}_2}} < \varepsilon
\end{equation}
is minimized. Here $\varepsilon$ is a certain positive performance scalar. Let us denote the minimum and maximum eigenvalues of $K_{dc}$ by $\lambda_{min}(K_{dc})$ and $\lambda_{max}(K_{dc})$, respectively, the Euclidean norm of $e_{obs}(t)$ by $\begin{Vmatrix}e_{obs}(t)\end{Vmatrix}$, and define $\alpha$ as a positive scalar. Observer gains are designed through the following theorem:
\begin{theorem}[] \label{th:thm1}
	Consider the error dynamics \eqref{eq:eo1}. If there exist a symmetric, positive definite matrix $K_{dc}\in\mathbb{R}^{2\times  2}$ and a real matrix $L_{dc}\in\mathbb{R}^{2\times 1}$ such that the following convex optimization problem is feasible:
	\begin{align}
	min~\nu~s.t.	\nonumber
	\end{align}
	\begin{equation}
	\begin{bmatrix}
	\Phi & K_{dc}B_{\xi} \\
	B_{\xi}^TK_{dc} & -\nu
	\end{bmatrix} < 0, \\
	\label{eq:dclmi1}
	\end{equation}
	\[\Phi=A_{dc}^TK_{dc}+K_{dc}A_{dc}-C_{dc}^TL_{dc}^T-L_{dc}C_{dc}+C_{dc}^TC_{dc}+2\alpha K_{dc},\]
	with $\nu=\varepsilon^2$, then the error dynamics \eqref{eq:eo1} are exponentially stable satisfying 
	\begin{equation}
	\begin{Vmatrix}e_{obs}(t)\end{Vmatrix}\leq\sqrt{\frac{\lambda_{max}(K_{dc})}{\lambda_{min}(K_{dc})}}e^{-\alpha t}\begin{Vmatrix}e_{obs}(0)\end{Vmatrix}
	\label{eq:eonorm}
	\end{equation}
when $g(t)=0$, and satisfy the $\mathcal{L}_2$ gain performance objective \eqref{eq:L2_gain} when $g(t)\neq 0$.
\end{theorem}
\begin{proof}
Using the definition of $\mathcal{L}_2$ gain, inequality \eqref{eq:L2_gain} can be rewritten after squaring both sides and some mathematical manipulation, as
	\begin{equation}
	\int_{0}^{\infty}(e_{obs}(t)^TC_{dc}^TC_{dc}e_{obs}(t)-\varepsilon^2g(t)^Tg(t))<0, 
	\label{eq:L2_gain2}
	\end{equation}
on which we impose the condition
	\begin{equation}
	\int_{0}^{\infty}(e_{obs}(t)^TC_{dc}^TC_{dc}e_{obs}(t)-\varepsilon^2g(t)^Tg(t))<-V_{obs}(e_{obs}(t)) 
	\label{eq:L2_gain3}
	\end{equation}
where $V_{obs}(e_{obs}(t))$ is a Lyapupov function for the observation error dynamics defined as 
	\begin{equation}
V_{obs}(e_{obs}(t))=e_{obs}(t)^TK_{dc}e_{obs}(t).
\label{eq:Vo1}
\end{equation}
Letting $V_{obs}(0)=0$, \eqref{eq:L2_gain3} can be cast as
	\begin{equation}
	\int_{0}^{\infty}(e_{obs}(t)^TC_{dc}^TC_{dc}e_{obs}(t)-\varepsilon^2g(t)^Tg(t)+\dot{V}_{obs}(e_{obs}(t)))dt<0 
	\label{eq:L2_gain4}
	\end{equation}
	which holds if 
	\begin{equation}
	e_{obs}(t)^TC_{dc}^TC_{dc}e_{obs}(t)-\varepsilon^2g(t)^Tg(t)+\dot{V}_{obs}(e_{obs}(t))<0 
	\label{eq:L2_gain5}
	\end{equation}
holds. Moreover, if the following inequality is satisfied then \eqref{eq:L2_gain5} also holds:
	\begin{equation}
	\begin{aligned}
	e_{obs}(t)^TC_{dc}^TC_{dc}e_{obs}(t)-\varepsilon^2g(t)^Tg(t)+\dot{V}_{obs}(e_{obs}(t))+2\alpha V_{obs}(e_{obs}(t))<0.
	\end{aligned}
	\label{eq:L2_gain6}
	\end{equation}
Differentiating $V_{obs}(e_{obs}(t))$ and simplifying gives 
\begin{equation}
\begin{aligned}
\dot{V}_{obs}(e_{obs}(t))&=e_{obs}(t)^T(A_{dc}^TK_{dc}+K_{dc}A_{dc}-C_{dc}^TL_{dc}^T \\
&-L_{dc}C_{dc})e_{obs}(t) +e_{obs}(t)^TK_{dc}B_{\xi}g(t)\\
&\qquad+g(t)^TB_{\xi}^TK_{dc}e_{obs}(t).
\end{aligned}
\label{eq:Vo2}
\end{equation}
Substituting \eqref{eq:Vo1} and \eqref{eq:Vo2} in \eqref{eq:L2_gain6} and writing it in matrix form leads to
	\begin{equation}
	\begin{bmatrix}e_{obs}(t)\\ g(t)\end{bmatrix}^T
	\begin{bmatrix}\Phi&K_{dc}B_{\xi}\\B_{\xi}^TK_{dc}&-\varepsilon^2\end{bmatrix}
	\begin{bmatrix}e_{obs}(t)\\ g(t)\end{bmatrix}<0,\\
	\label{eq:L2_gain7}
	\end{equation}
where $\Phi$ is defined in \eqref{eq:dclmi1}. Substituting $\nu=\varepsilon^2$, a sufficient condition for \eqref{eq:L2_gain7} to hold is 
	\begin{equation}
	\begin{bmatrix}\Phi&K_{dc}B_{\xi}\\B_{\xi}^TK_{dc}&-\nu\end{bmatrix}<0
	\label{eq:L2_gain8}
	\end{equation}
which proves \eqref{eq:dclmi1} in Theorem \ref{th:thm1}, and hence satisfies the $\mathcal{L}_2$ gain performance constraint \eqref{eq:L2_gain}. To fulfill the exponential stability requirement \eqref{eq:eonorm} on the $e_{obs}(t)$ dynamics, we seek the sufficient conditions to have 
\begin{equation}
\dot{V}_{obs}(e_{obs}(t)) \leq -2\alpha V_{obs}(e_{obs}(t))
\label{eq:Vo3}
\end{equation} 
because \eqref{eq:Vo3} implies the following set of inequalities:
\begin{equation}
\begin{aligned}
V_{obs}(e_{obs}(t))&\leq e^{-2\alpha t}V_{obs}(0),\\
\lambda_{min}(K_{dc})\begin{Vmatrix}e_{obs}(t)\end{Vmatrix}^2&\leq V_{obs}(e_{obs}(t)) \leq e^{-2\alpha t}V_{obs}(0) \\ &\qquad\qquad\leq \lambda_{max}(K_{dc})e^{-2\alpha t}\begin{Vmatrix}e_{obs}(0)\end{Vmatrix}^2.
\end{aligned}
\label{eq:Vo6}
\end{equation}
From the second inequality in \eqref{eq:Vo6}, we can derive
\begin{equation}
\begin{Vmatrix}e_{obs}(t)\end{Vmatrix}\leq\sqrt{\frac{\lambda_{max}(K_{dc})}{\lambda_{min}(K_{dc})}}e^{-\alpha t}\begin{Vmatrix}e_{obs}(0)\end{Vmatrix}
\label{eq:Vo7}
\end{equation}
showing that \eqref{eq:eonorm} in Theorem \ref{th:thm1} holds. Substituting \eqref{eq:Vo1} and \eqref{eq:Vo2} into \eqref{eq:Vo3} leads to
\begin{equation}
\begin{aligned}
&e_{obs}(t)^T(A_{dc}^TK_{dc}+K_{dc}A_{dc}-C_{dc}^TL_{dc}^T-L_{dc}C_{dc}+ 2\alpha K_{dc})e_{obs}(t) \leq 0
\end{aligned}
\label{eq:Vo4}
\end{equation}
when $g(t)=0$. Clearly, \eqref{eq:Vo4} holds if 
\begin{equation}
A_{dc}^TK_{dc}+K_{dc}A_{dc}-C_{dc}^TL_{dc}^T-L_{dc}C_{dc}+2\alpha K_{dc} \leq 0
\label{eq:Vo5}
\end{equation}
holds. Now, when $g(t)=0$ in \eqref{eq:L2_gain7}, it reduces to $e_{obs}(t)^T\Phi e_{obs}(t)<0$, for which it is sufficient that $\Phi<0$, implying that \eqref{eq:Vo5} holds since $C_{dc}^TC_{dc}\geq 0$, and hence fulfilling the exponential stability requirement on \eqref{eq:eo1}. This completes the proof.
\end{proof}
The control law $u(t)$ is designed such that the disturbances affecting the $v_{dc}(t)$ dynamics are estimated using observer \eqref{eq:vdcobs1} and canceled. To this end, we set 
\begin{equation}
u(t)=k_{dc}(v_{dc}(t)-v_{dc}^{*}(t))+\hat{\xi}_{dc}(t),
\label{eq:controlvdc}
\end{equation}
where $v_{dc}^*(t)$ is the reference signal for $v_{dc}(t)$ and $k_{dc}$ is a positive  design constant to be selected. It is worth noting that \eqref{eq:controlvdc} is a simple proportional controller with an added disturbance compensation term $\hat{\xi}_{dc}(t)$. 

\subsection{Current Control} 			\label{ssec:currentcontrol}
Parameters $R$ and $L$ are not accurately known but are assumed to lie in a range 
\begin{equation}
\begin{aligned}
R_{min}\leq R \leq R_{max},\\
L_{min}\leq L \leq L_{max}.
\end{aligned}
\label{eq:RLbounds}
\end{equation}
Omitting the $\alpha\beta$ subscripts in \eqref{eq:model2}, grid current dynamics can be written in a general form as
\begin{equation}
\frac{di(t)}{dt}=-\frac{R}{L}i(t) + \frac{1}{L}u_t(t) - \frac{1}{L}V_{grid}(t) + \zeta(t)
\label{eq:igeneric}
\end{equation}
where $i(t)$ is the state variable, $u_t(t)$ is the control input, $V_{grid}(t)$ is the grid voltage and $\zeta(t)$ denotes all the disturbances and unmodeled effects. Treating the grid voltage $V_{grid}(t)$ as a disturbance, \eqref{eq:igeneric} can be rewritten as
\begin{equation}
\frac{di(t)}{dt}=-\frac{R}{L}i(t) + \frac{1}{L}u_t(t) + \xi(t)
\label{eq:igeneric2}
\end{equation}
with $\xi(t)=-\frac{1}{L}V_{grid}(t)+\zeta(t)$ as the new disturbance. Let
\begin{equation}
\begin{aligned}
\rho_1&=\frac{R}{L},~\rho_{1,min}=\frac{R_{min}}{L_{max}},~\rho_{1,max}=\frac{R_{max}}{L_{min}},\\
\rho_2&=\frac{1}{L},~\rho_{2,min}=\frac{1}{L_{max}},~\rho_{2,max}=\frac{1}{L_{min}}.
\end{aligned}
\label{eq:thetadef1}
\end{equation}
We then have
\begin{equation}
\begin{aligned}
&\rho_{1,min}\leq\rho_1\leq\rho_{1,max},\\
&\rho_{2,min}\leq\rho_2\leq\rho_{2,max}.\\
\end{aligned}
\label{eq:thetadef2}
\end{equation}
In matrix form, \eqref{eq:igeneric2} can be written  as
\begin{equation}
\begin{aligned}
\dot{i}(t)&=A_u(\rho)i(t)+B_u(\rho)u_t(t)+B_{du}\xi(t),\\
y(t)&=C_ui(t),
\end{aligned}
\label{eq:igeneric3}
\end{equation}
with $A_u(\rho)=-\rho_1$, $B_u(\rho)=\rho_2$, $B_{du}=1$ and $C_u=1$. Uncertainties $A_u(\rho)$ and $B_u(\rho)$ can be written in the form
\begin{equation}
A_u(\rho)=\sum_{i=1}^{\mu}\eta_{i}A_{u,i},~B_u(\rho)=\sum_{i=1}^{\mu}\eta_{i}B_{u,i}
\label{eq:apbpdef}
\end{equation}
for $\eta=\begin{bmatrix}\eta_{1}&\eta_{2}&\cdots&\eta_{\mu}\end{bmatrix}^T\in\mathbb{R}^{\mu},$ an uncertain constant parameter vector satisfying
\begin{equation}
\eta\in\Omega:=\Bigg\{\eta\in\mathbb{R}^{\mu}:\eta_{i}\geq0~(i=1,\cdots,\mu),~\sum_{i=1}^{\mu}\eta_{i}=1 \Bigg\}.
\label{eq:uncvec}
\end{equation}
We design a low-pass filter based repetitive controller (RC) shown in Fig. \ref{fig:lpfrc} for tracking the grid currents, based on the results presented in \cite{Bonan2011}, but alter them to treat uncertainties in the input matrix that were not treated therein. The low-pass filter ensures closed loop stability. 
\begin{figure}[b!]
	\centering
	\includegraphics[scale=1]{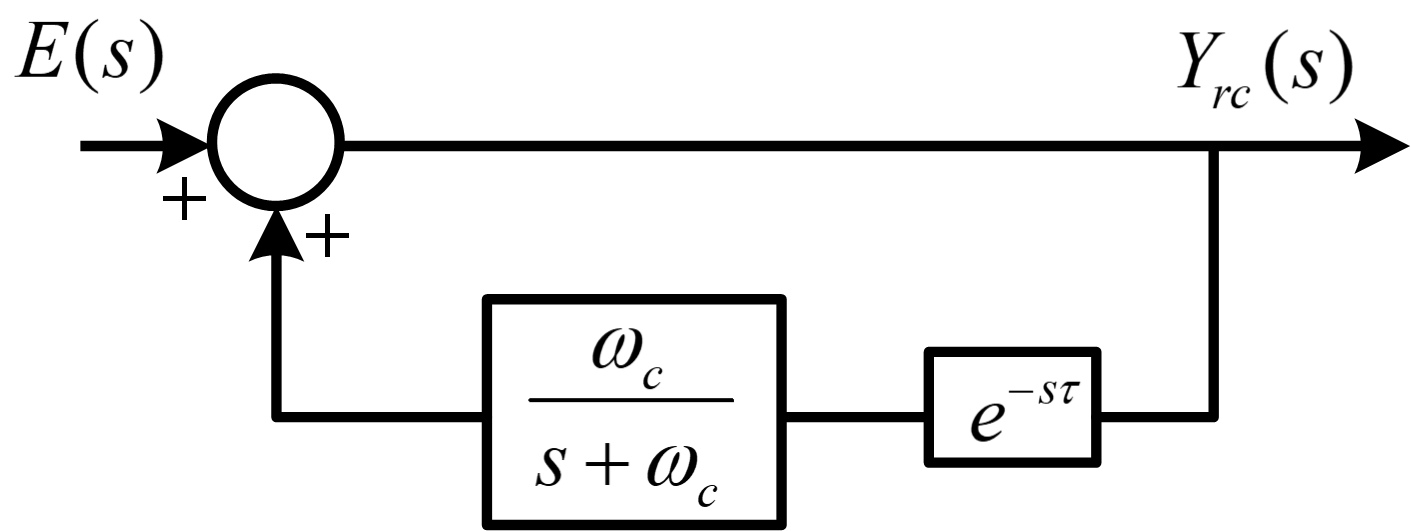}
	\caption{Repetitive controller with a low-pass filter.}
	\label{fig:lpfrc}
\end{figure}
The closed loop form of the controller in Fig. \ref{fig:lpfrc} can be realized in state-space as \cite{Bonan2011}
\begin{equation}
\begin{aligned}
\dot{x}_{rc}(t)&=-\omega_cx_{rc}(t)+\omega_cx_{rc}(t-\tau)+\omega_ce(t-\tau),\\
y_{rc}(t)&=x_{rc}(t) + e(t),
\end{aligned}
\label{eq:ssrc}
\end{equation}
where $x_{rc}(t)$ is the filter state, $\omega_c$ is the low-pass filter cutoff frequency ($\mathrm{rad/s}$), $\tau$ is a time delay equal to the inverse of grid frequency ($\mathrm{s}$), and $e(t)$ is the current tracking error defined as $e(t)=i^{*}(t)-C_ui(t)$ with $i^{*}(t)$ as the reference for current. $E(s)$ and $Y_{rc}(s)$ denote the Laplace transforms of $e(t)$ and $y_{rc}(t)$, respectively.  Defining a new state vector as $x(t)=\begin{bmatrix}i(t)&x_{rc}(t)\end{bmatrix}^T$, the augmented system can be written as
\begin{equation}
\begin{aligned}
\dot{x}(t)&=A(\rho)x(t)+A_dx(t-\tau)+B(\rho)u_t(t)\\
&+B_d\xi(t)+\bar{B}_ri^{*}(t-\tau)
\end{aligned}
\label{eq:augsys}
\end{equation}
where
\begin{equation}
\begin{aligned}
x(t)&=\begin{bmatrix}i(t)\\x_{rc}(t)\end{bmatrix},~
A(\rho)=\begin{bmatrix}A_u(\rho)&0\\0&-\omega_c\end{bmatrix},\\
B(\rho)&=\begin{bmatrix}B_u(\rho)\\0\end{bmatrix},~
A_d=\begin{bmatrix}0&0\\-C_u\omega_c&\omega_c\end{bmatrix},\\
B_d&=\begin{bmatrix}B_{du}\\0\end{bmatrix},~
\bar{B}_r=\begin{bmatrix}0\\ \omega_c\end{bmatrix}.~
\end{aligned}
\label{eq:augsysmats}
\end{equation}
Uncertain matrices $A(\rho)$ and $B(\rho)$ can be expressed as 
\begin{equation}
A(\rho)=\sum_{i=1}^{\mu}\eta_{i}A_{i},~B(\rho)=\sum_{i=1}^{\mu}\eta_{i}B_{i}.
\label{eq:atbtdef}
\end{equation}
We select a state feedback of the form
\begin{equation}
u_t(t)=k_1i(t)+k_2y_{rc}(t)
\label{eq:statefeedback}
\end{equation}
using some constants $k_1$ and $k_2$ such that \eqref{eq:augsys} is globally asymptotically stable, and the following transient performance cost function is minimized:
\begin{equation}
J(p(t))\coloneqq \begin{Vmatrix}p(t)\end{Vmatrix}^{2}=\int_{0}^{\infty}p(t)^Tp(t)dt,
\label{eq:costfcn}
\end{equation}
where $\begin{Vmatrix}p(t)\end{Vmatrix}$ represents the Eulcidean norm of $p(t)$, a controlled output function of state $x(t)$ and input $u_t(t)$ defined as
\begin{equation}
p(t)\coloneqq Gx(t)+Hu_{t}(t),
\end{equation}
where $G$ and $H$ are performance scaling matrices of appropriate dimensions. Equation \eqref{eq:statefeedback} can be written in terms of $x(t)$ as 
\begin{equation}
u_t(t)=Fx(t)+k_2i^{*}(t)
\label{eq:statefeedbackaug}
\end{equation}
leading to the closed loop dynamics
\begin{equation}
\begin{aligned}
\dot{x}(t)&=(A(\rho)+B(\rho)F)x(t)+A_dx(t-\tau)+B_d\xi(t)\\
&+B_r(\rho)r(t)
\end{aligned}
\label{eq:augsyscl}
\end{equation}
with
\begin{equation}
\begin{aligned}
F&=\begin{bmatrix}(k_1-k_2C_u)\\ k_2\end{bmatrix}^T,~
B_r(\rho)=\begin{bmatrix}B(\rho)k_2&\bar{B}_r\end{bmatrix},\\
r(t)&=\begin{bmatrix}i^{*}(t)\\ i^{*}(t-\tau)\end{bmatrix}.
\end{aligned}
\label{eq:auxdefs}
\end{equation}
According to the internal model principal \cite{Yamamoto1988,Francis1975}, if the feedback assures internal asymptotic stability for \eqref{eq:augsyscl}, then it will also asymptotically track references and reject disturbances. Hence, for internal stability of \eqref{eq:augsyscl}, we ignore the signals $\xi(t)$ and $r(t)$ and seek to robustly stabilize 
\begin{equation}
\dot{x}(t)=(A(\rho)+B(\rho)F)x(t)+A_dx(t-\tau)
\label{eq:augsysinternal}
\end{equation}
through a state feedback designed in the following theorem:
\begin{theorem}[] \label{th:th2}
	If there exist two symmetric, positive definite matrices $X\in\mathbb{R}^{2\times 2}$ and $W\in\mathbb{R}^{2\times 2}$, a real matrix $Y\in\mathbb{R}^{1\times 2}$ and a positive scalar $\gamma$ such that the following LMI is feasible:
	\begin{equation}
	\begin{aligned}
	&\begin{bmatrix}\Theta_{i}&A_dX&\Gamma\\XA_d^T&-W&\mathbf{0}\\ \Gamma^{T}&\mathbf{0}&-\gamma I\end{bmatrix} < 0,~\forall~i=1,2,\cdots,\mu,\\
	&\Theta_i=A_iX+XA_i^T+B_iY+Y^TB_i^T+W,\\
	&\Gamma=XG^{T}+Y^{T}H^{T},
	\end{aligned}
	\label{eq:ilmi}
	\end{equation}
	then the closed loop system \eqref{eq:augsysinternal} with $F=YX^{-1}$ is robustly asymptotically stable and the cost function \eqref{eq:costfcn} satisfies $\begin{Vmatrix}p(t)\end{Vmatrix}^2\leq \gamma V(x(t)=0)$ where
	\begin{equation}
	V(0)=x(0)^TX^{-1}x(0) + \int_{-\tau}^{0}x(\vartheta)X^{-1}S^{-1}X^{-1}x(\vartheta)d\vartheta.
	\label{lyapkrav0}
	\end{equation}
\end{theorem}
\begin{proof}
	Define a Lyapunov-Krasovskii functional as
	\begin{equation}
	V(x(t))=x(t)^TX^{-1}x(t) + \int_{t-\tau}^{t}x(\vartheta)Qx(\vartheta)d\vartheta,
	\label{eq:lyapkrav}
	\end{equation}
	where $X^{-1}\in\mathbb{R}^{2\times 2}$ and $Q\in\mathbb{R}^{2\times 2}$ are two symmetric, positive definite matrices. Differentiating \eqref{eq:lyapkrav} gives
	\begin{equation}
	\begin{aligned}
	\dot{V}(x(t))&=x(t)^T\Big(X^{-1}A(\rho)+A(\rho)^TX^{-1}+X^{-1}B(\rho)F+F^TB(\rho)^TX^{-1}\\&\qquad+Q\Big)x(t)+x(t)^TX^{-1}A_dx(t-\tau)+x(t-\tau)^TA_d^TX^{-1}x(t)\\
	&-x(t-\tau)^TQx(t-\tau)
	\end{aligned}
	\label{eq:vdotlk1}
	\end{equation}
	or 
	\begin{equation}
	\begin{aligned}
	\dot{V}(x(t))&=\begin{bmatrix}x(t)\\x(t-\tau)\end{bmatrix}^T
	\begin{bmatrix}\Psi_1&X^{-1}A_d\\A_d^TX^{-1}&-Q\end{bmatrix}
	\begin{bmatrix}x(t)\\x(t-\tau)\end{bmatrix},\\
	\Psi_1&=X^{-1}A(\rho)+A(\rho)^TX^{-1}+X^{-1}B(\rho)F+F^TB(\rho)^TX^{-1}+Q,
	\end{aligned}
	\label{eq:vdotlk2}
	\end{equation}
	in matrix form. According to the Lyapunov-Krasovskii theory \cite{Gu2003}, if 
	\begin{equation}
	\dot{V}(x(t))+\gamma^{-1}p(t)^{T}p(t)<0
	\label{eq:LyapKrav1}
	\end{equation}
	holds, then \eqref{eq:augsysinternal} is asymptotically stable with $\begin{Vmatrix}
	p(t)\end{Vmatrix}^2\leq\gamma V(0)$. Inequality \eqref{eq:LyapKrav1} can be written in matrix form as 
	\begin{equation}
	\begin{aligned}
	&\begin{bmatrix}x(t)\\x(t-\tau)\end{bmatrix}^T
	\begin{bmatrix}\Psi_2&X^{-1}A_d\\A_d^TX^{-1}&-Q\end{bmatrix}
	\begin{bmatrix}x(t)\\x(t-\tau)\end{bmatrix}<0,\\
	\Psi_2&=X^{-1}A(\rho)+A(\rho)^TX^{-1}+X^{-1}B(\rho)F+F^TB(\rho)^TX^{-1}\\
	&+Q+\gamma^{-1}(G+HF)^{T}(G+HF),
	\end{aligned}
	\label{eq:vdotlk2a}
	\end{equation}
	for which it is sufficient that
	\begin{equation}
	\begin{bmatrix}\Psi_2&X^{-1}A_d\\A_d^TX^{-1}&-Q\end{bmatrix}<0.
	\label{eq:vdotlk3}
	\end{equation}
	Let $F=YX^{-1}$ and $W=XQX$. Multiplying \eqref{eq:vdotlk3} on both sides by $\begin{bmatrix}X&\mathbf{0}\\\mathbf{0}&X\end{bmatrix}$, we get
	\begin{equation}
	\begin{bmatrix}\Psi_3&A_dX\\XA_d^T&-W\end{bmatrix}<0
	\label{eq:vdotlk4}
	\end{equation}
	where
	\begin{equation}
	\begin{aligned}
	\Psi_3&=A(\rho)X+XA(\rho)^T+B(\rho)Y+Y^TB(\rho)+W\\
	&\qquad+\gamma^{-1}(GX+HY)^T(GX+HY).
	\end{aligned}
	\label{eq:Psi2def}
	\end{equation}
Using \eqref{eq:atbtdef} and the Schur's complement, we can write \eqref{eq:vdotlk4} as 
	\begin{equation}
	\sum_{i=1}^{\mu}\eta_{i}\begin{bmatrix}\Theta_i&A_dX&\Gamma\\XA_d^T&-W&\mathbf{0}\\ \Gamma^{T}&\mathbf{0}&-\gamma I\end{bmatrix}<0.
	\label{eq:vdotlk5}
	\end{equation}
where $\Theta_{i}$ and $\Gamma$ are defined in \eqref{eq:ilmi}. Since $\sum_{i=1}^{\mu}\eta_{i}=1$, \eqref{eq:vdotlk5} holds if
	\begin{equation}
	\begin{bmatrix}\Theta_i&A_dX&\Gamma\\XA_d^T&-W&\mathbf{0}\\ \Gamma^{T}&\mathbf{0}&-\gamma I\end{bmatrix}<0,~ \forall i=1,\cdots,\mu
	\label{eq:vdotlk6}
	\end{equation}
holds. This proves Theorem \ref{th:th2} and completes the proof.
\end{proof}
\begin{remark}
An exponential convergence rate can be imposed on \eqref{eq:augsysinternal} so that 
	\begin{equation}
	\begin{Vmatrix}x(t)\end{Vmatrix}\leq \beta\begin{Vmatrix}x(0)\end{Vmatrix}e^{-\lambda t}.
	\label{eq:expdecay}
	\end{equation}
for some positive constants $\lambda$ and $\beta$ if the following LMI is solvable:
	\begin{equation}
	\begin{bmatrix}\Theta_i+2\lambda X&A_dX&\Gamma\\XA_d^T&-W&\mathbf{0}\\ \Gamma^{T}&\mathbf{0}&-\gamma I\end{bmatrix}<0,~ \forall i=1,\cdots,\mu
	\label{eq:expdecaylmi}
	\end{equation}
It can be shown from the proof by using the state transformation $\phi(t)=e^{\lambda t}x(t)$. See \cite{Bonan2011} for details.
\end{remark}
\begin{remark}
Both the voltage and current controllers are designed independent of the PV or FC generator dynamics. This makes the control scheme more robust and immune to the nonlinearities and disturbances present in those dynamics.
\end{remark}

\subsection{Selecting Parameters} \label{ssec:parachoice}
Define $e_{dc}(t)=v_{dc}(t)-v_{dc}^{*}(t)$ as the voltage tracking error and let $v_{dc}^{*}(t)$ be a constant. Once the observer \eqref{eq:vdcobs1} state converges, the disturbance $\xi(t)$ is perfectly estimated and canceled by the control law  \eqref{eq:controlvdc}. Error $e_{dc}(t)$ dynamics can then be written by using \eqref{eq:vdcmodel3} as
\begin{equation}
\dot{e}_{dc}(t)=-k_{dc}e_{dc}(t)
\label{eq:edcdot}
\end{equation} 
which implies that 
\begin{equation}
{e}_{dc}(t)=e_{dc}(0)e^{-k_{dc}t}.
\label{eq:edc}
\end{equation} 
Hence, $\lambda$ and $k_{dc}$ dictate the convergence rates of the dc voltage and phase current closed loop dynamics, respectively, as seen in \eqref{eq:expdecay} and \eqref{eq:edc}. Dynamics of current should at least be five times faster than dynamics of the dc-link voltage to prevent them from interacting with each other \cite{Mirhosseini2016}. The current loop's time constant, denoted by $\tau_i$, is normally chosen between $0.5\text{-}5~\mathrm{ms}$ \cite{Yazdani2010}. Therefore, $\lambda$ and $k_{dc}$ can be chosen as 
\begin{equation}
\lambda=\frac{1}{\tau_i},~k_{dc}=\frac{\lambda}{5}.
\end{equation}
For the low-pass filter, a small value of $\omega_c$ yields a good transient response but compromises steady state performance \cite{Bonan2011}. Since we're controlling the transient responses of dc-link voltage and grid currents through the cost function \eqref{eq:costfcn} as well as the parameters $\lambda$ and $k_{dc}$, a large $\omega_c$ is chosen to gain good steady state performance.

\subsection{Control of PV and FC Generators}	\label{pvfc_control}
The PV generator's boost converter is controlled by maximum power point tracking (MPPT) control using the perturb and observe algorithm under normal conditions. However, under voltage sags, in order to prevent the VSC currents from exceeding their limits, PV generator's power production is curtailed by controlling its dc-dc converter using the $v_{dc}(t)$ controller through $i_{dc}^{*}(t)$ as per \eqref{eq:idc1}, and subsequently generating the PWM pulses. Thus, regulating the dc-link voltage through the PV generator's boost converter keeps it stable during voltage sags. The FC generator is controlled via a simple proportional controller that acts on the difference between $P_{fc}^{*}(t)$ and $P_{fc}(t)$ and generates the switching pulses via PWM. The proportional gain is selected as $50$ to obtain the best response. 

\subsection{Controller Reference Signal Generation} 	\label{ssec:refgen}
Reference for $v_{dc}(t)$ is independently chosen but must be sufficiently high for proper operation of VSC \cite{Yazdani2010}. Currents are controlled in the stationary $\alpha\beta$ reference frame. Their references $i_{\alpha}^{*}(t)$ and $i_{\beta}^{*}(t)$ are generated using power references $P_{grid}^*(t)$ and $Q_{grid}^*(t)$, and voltages $V_{grid,\alpha}(t)$ and $V_{grid,\beta}(t)$. Under normal conditions, $P_{grid}^*(t)$ comes from the $v_{dc}(t)$ controller as per \eqref{eq:Pgrid1} and $Q_{grid}^{*}(t)$ comes from the EMS. The current references $i_{\alpha}^{*}(t)$ and $i_{\beta}^{*}(t)$ are calculated as \cite{Yazdani2010}
\begin{equation}
\begin{aligned}
i_{\alpha}^{*}(t)&=\frac{2}{3}\frac{V_{grid,\alpha}(t)P_{grid}^*(t)+V_{grid,\beta}(t)Q_{grid}^*(t)}{V_{grid,\alpha}^2(t)+V_{grid,\beta}^2(t)},\\
i_{\beta}^{*}(t)&=\frac{2}{3}\frac{V_{grid,\beta}(t)P_{grid}^*(t)-V_{grid,\alpha}(t)Q_{grid}^*(t)}{V_{grid,\alpha}^2(t)+V_{grid,\beta}^2(t)}.
\end{aligned}
\label{eq:irefsnormal}
\end{equation}
Under voltage sag conditions, both $P_{grid}^*(t)$ and $Q_{grid}^*(t)$ come from the EMS and current references $i_{\alpha}^{*}(t)$ and $i_{\beta}^{*}(t)$ are calculated using the delayed voltage control method \cite{Liu2016a} as follows:
\begin{equation}
\begin{aligned}
i_{\alpha}^{*}(t)&=\frac{2}{3}\frac{V_{grid,\beta}(t)Q_{grid}^*(t)-\tilde{V}_{grid,\beta}(t)P_{grid}^*(t)}{V_{grid,\beta}(t)\tilde{V}_{grid,\alpha}(t)-V_{grid,\alpha}(t)\tilde{V}_{grid,\beta}(t)},\\
i_{\beta}^{*}(t)&=\frac{2}{3}\frac{\tilde{V}_{grid,\alpha}(t)P_{grid}^*(t)-V_{grid,\alpha}(t)Q_{grid}^*(t)}{V_{grid,\beta}(t)\tilde{V}_{grid\alpha}(t)-V_{grid\alpha}(t)\tilde{V}_{grid\beta}(t)},
\end{aligned}
\label{eq:irefsfault}
\end{equation}
where $\tilde{V}_{grid,\alpha}(t)$ and $\tilde{V}_{grid,\beta}(t)$ are the delayed versions of $V_{grid,\alpha}(t)$ and $V_{grid,\beta}(t)$, delayed by $\tau/4$ seconds. This technique allows for a constant real power injection and sinusoidal currents waveforms during asymmetrical voltage sags. As per \eqref{eq:statefeedback}, control laws take the final form
\begin{equation}
\begin{aligned}
u_{t\alpha}&=k_1i_{\alpha}(t)+k_2(i_{\alpha}^{*}(t)-C_{u}i_{\alpha}(t)+x_{rc,\alpha}(t)),\\
u_{t\beta}&=k_1i_{\beta}(t)+k_2(i_{\beta}^{*}(t)-C_{u}i_{\beta}(t)+x_{rc,\beta}(t)),
\end{aligned}
\label{eq:icontrollaw}
\end{equation}
where $x_{rc,\alpha}(t)$ and $x_{rc,\beta}(t)$ denote the RC states for $\alpha$ and $\beta$ current control loops, respectively. Block diagram of the control system is shown in Fig. \ref{fig:controlblockdia}. Only one of the two FC dc-dc converters' power control loop is shown due to space limitation.  
\begin{figure}[t!]
	\centering
	\includegraphics[scale=0.5]{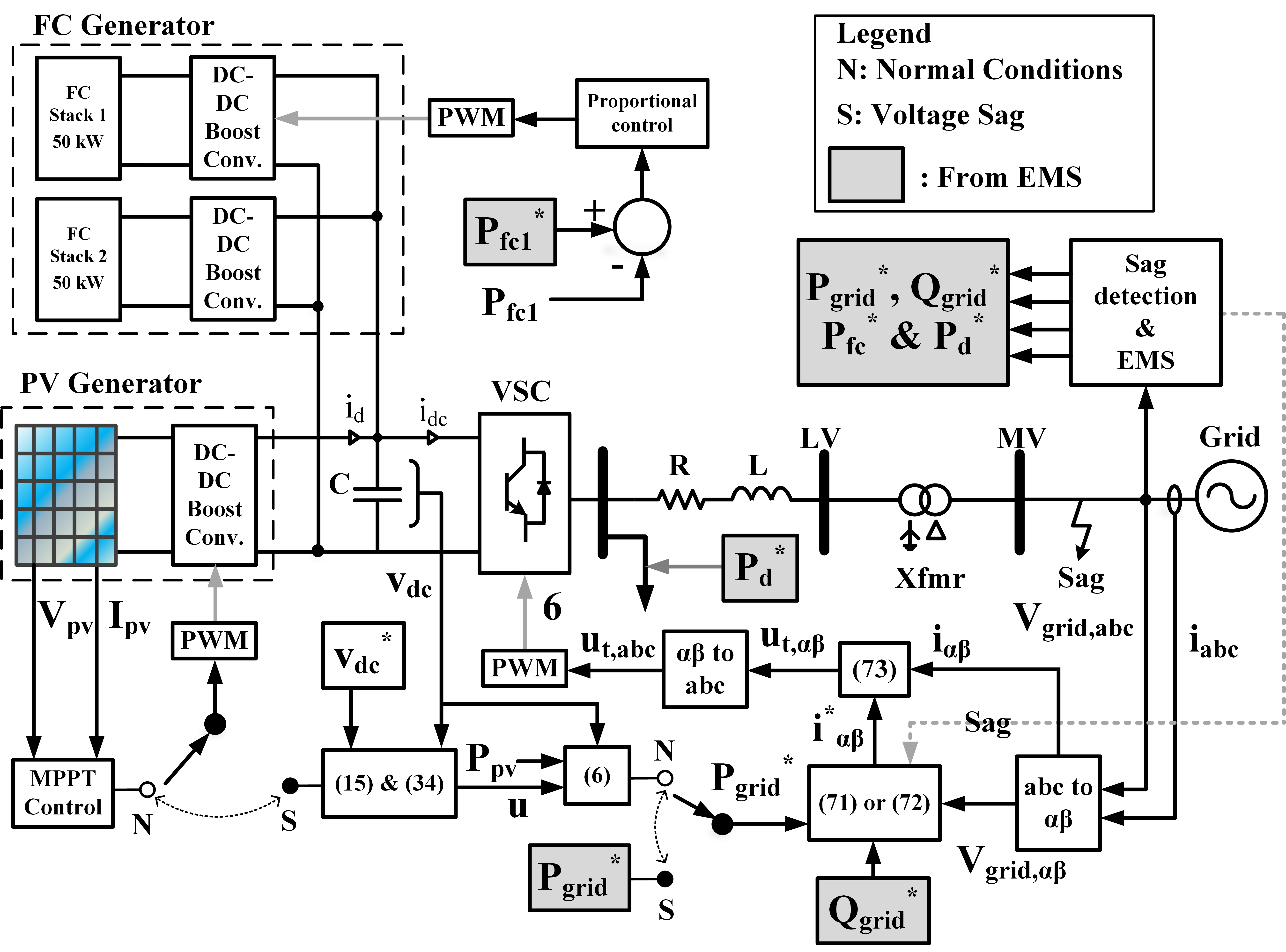}
	\caption{Control system block diagram.}
	\label{fig:controlblockdia}
\end{figure}

\section{Simulation Results and Discussion} \label{sec:simulations}
Controller gains are designed by solving LMIs \eqref{eq:dclmi1} and \eqref{eq:ilmi} using the Yalmip toolbox \cite{Lofberg2004} in MATLAB (MathWorks, Natick, MA). LMI \eqref{eq:ilmi} is solved with uncertainty limits of $\pm 30\%$ around the nominal values of $R$ and $L$. Gain and parameter values are listed in Table \ref{tbl:controlgains}. Simulations are carried out in the SimPowerSystems\textsuperscript{TM} toolbox of MATLAB at a sampling frequency of $0.2~\mathrm{MHz}$. EMS and control system performance are tested under four different cases namely case 1: unity power factor (UPF) operation, case 2: real-reactive (PQ) power support, case 3: voltage sags under nominal irradiance, and case 4: voltage sags under low irradiance. $R$, $L$ and $C$ are set to $30\%$ above their nominal values. All test case are run for $10~\mathrm{s}$ duration and reference for $v_{dc}(t)$ is set to $800~\mathrm{V}$ throughout.
\begin{table}[!t]
	\centering
	\caption{Controller Parameters and Gains}
	\label{tbl:controlgains}
	\begin{tabular}{ l r }
		\toprule[0.5mm]
		\toprule[0.5mm]
		\multicolumn{2}{c}{DC Voltage Controller}\\
		\toprule
		$\varepsilon,~\alpha,~k_{dc}$ & $0.8405,~50,~100$\\
		$K_{dc},~L_{dc}$& $\begin{bmatrix}3.6043&-0.0359\\-0.0359&0.00007\end{bmatrix}$, $\begin{bmatrix}180.8163\\0.0157\end{bmatrix}$\\
		\midrule
		\multicolumn{2}{c}{Current Controller}\\ \toprule
		$k_1,~k_2,~\lambda,~\omega_c$&$-0.1649,~1.2197\times 10^4,~500,~1000$\\
		$F$& $\begin{bmatrix}-1.2197\\1.2197\end{bmatrix}\times 10^4$\\
		\toprule
		\toprule
	\end{tabular}	
\end{table}

\subsection{Case 1: UPF Operation} \label{ssec:upf}
In this test case, the reactive power demand $Q^{*}(t)$ is set to $0~\mathrm{kVAR}$ for UPF operation. At $t=0~\mathrm{s}$, the GO provides the following step-changing profile for the real power demand $P^{*}(t)$: $150~\mathrm{kW}$ from $t=0~\mathrm{s}$ to $t=2~\mathrm{s}$, $220~\mathrm{kW}$ from $t=2~\mathrm{s}$ to $t=4~\mathrm{s}$, $80~\mathrm{kW}$ from $t=4~\mathrm{s}$ to $t=6~\mathrm{s}$, and $150~\mathrm{kW}$ from $t=6~\mathrm{s}$ to $t=10~\mathrm{s}$. Plots of $P^{*}(t)$ and real power $P_{grid}(t)$ delivered to the grid are shown in Fig. \ref{fig:case1_P} (a). PV and FC power outputs namely $P_{pv}(t)$ and $P_{fc}(t)$ and the dump load power consumption $P_d(t)$ are shown in Fig. \ref{fig:case1_P} (b). Figure \ref{fig:case1_QVdc} shows the active power delivered to the $Q_{grid}(t)$, reactive power reference $Q^{*}(t)$, voltage $v_{dc}(t)$ and its reference $v_{dc}^{*}(t)$. The irradiance is $1000~\mathrm{W/m^2}$ except during $t=6~\mathrm{s}$ to $t=8~\mathrm{s}$, where it steps down to $300~\mathrm{W/m^2}$.
\begin{figure}[t!]
	\centering
	\includegraphics[trim={0 0 0 0},scale=0.65]{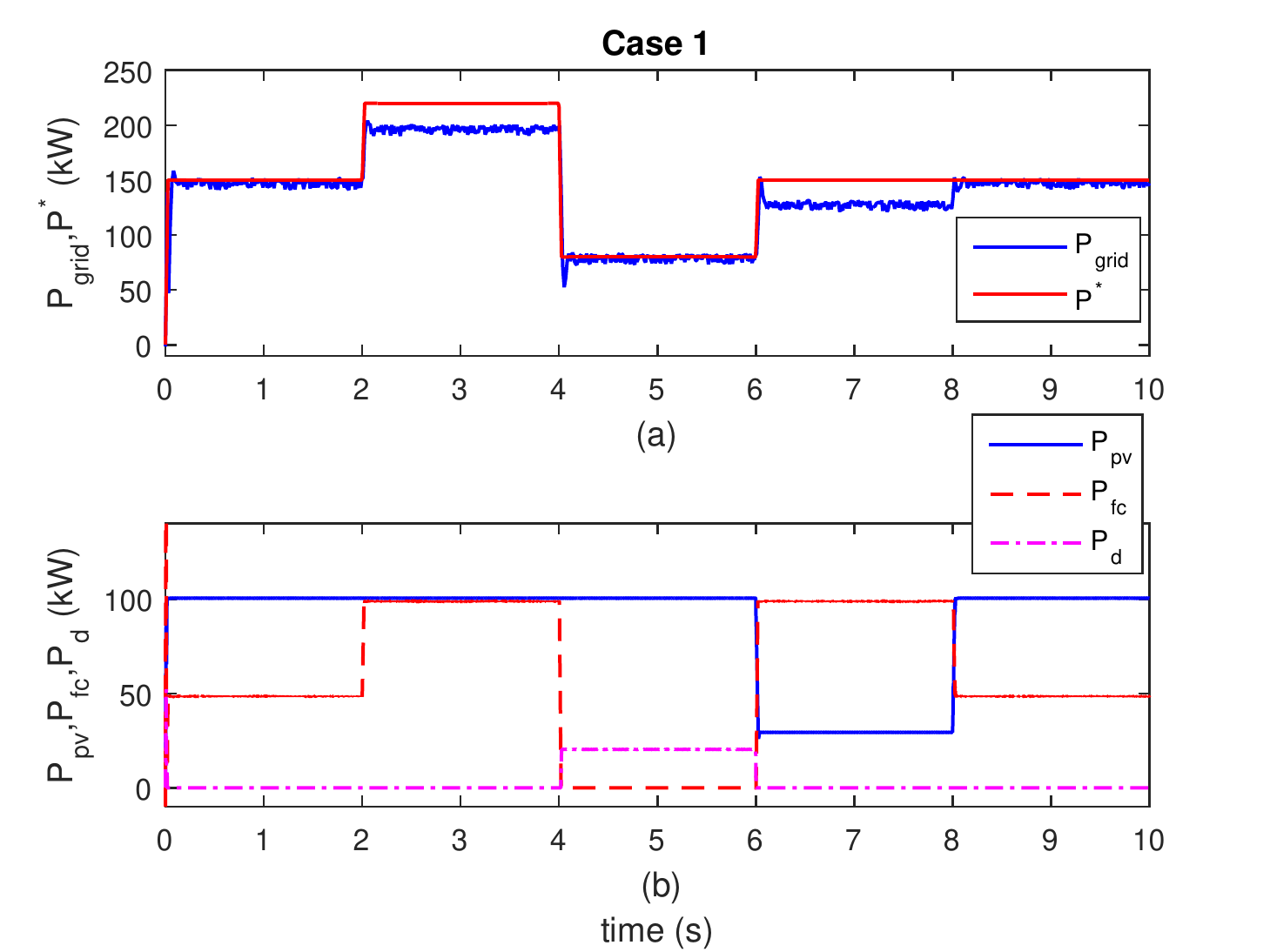}
	\caption{Case 1: UPF operation: (a) $P_{grid}(t)$ and $P^*(t)$, and (b) $P_{pv}(t)$, $P_{fc}(t)$ and $P_{d}(t)$.}
	\label{fig:case1_P}
\end{figure}
\begin{figure}[t!]
	\centering
	\includegraphics[trim={0 0 0 0},scale=0.65]{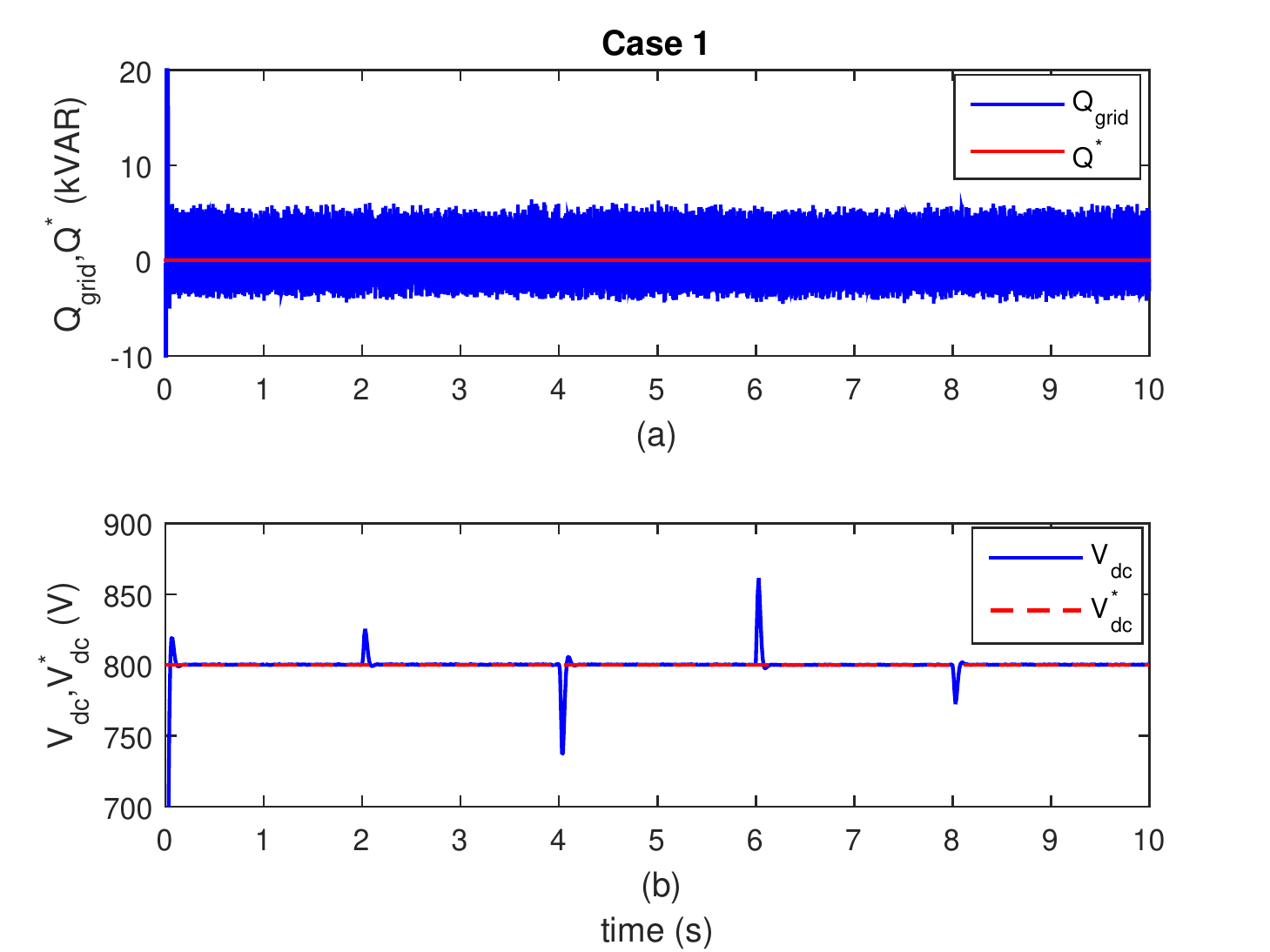}
	\caption{Case 1: UPF operation: (a) $Q_{grid}(t)$ and $Q^*(t)$, and (b) $v_{dc}(t)$ and $v_{dc}^{*}(t)$.}
	\label{fig:case1_QVdc}
\end{figure}

Figure \ref{fig:case1_P} (a) indicates that $P^{*}(t)$ is effectively tracked by the control system for the most part, with exceptions in the intervals $t=2~\mathrm{s}$ to $t=4~\mathrm{s}$ and $t=6~\mathrm{s}$ to $t=8~\mathrm{s}$ that are explained. From $t=2~\mathrm{s}$ to $t=4~\mathrm{s}$, the GO's demand of $150~\mathrm{kW}$ is completely met by the hybrid PV-FC system; $P_{pv}(t)$ is $100~\mathrm{kW}$ and the EMS assigns the $50~\mathrm{kW}$ deficit to the FC generator, as seen in Fig. \ref{fig:case1_P} (b). During $t=2~\mathrm{s}$ to $t=4~\mathrm{s}$, demand $P^{*}(t)$ increases to $220~\mathrm{kW}$ which is higher than the maximum system rating of $200~\mathrm{kW}$. Hence, the EMS sets the FC power reference to its rated value $P_{fc}^{rated}$ of $100~\mathrm{kW}$ and the system delivers all available PV power plus $P_{fc}^{rated}$ to the grid, amounting to $200~\mathrm{kW}$ and an unmet deficit of $20~\mathrm{kW}$ results. In the interval $t=4~\mathrm{s}$ to $t=6~\mathrm{s}$, $P^{*}(t)$ decreases to $80~\mathrm{kW}$. Since it is less than the available $P_{pv}(t)$ of $100~\mathrm{kW}$, a $20~\mathrm{kW}$ surplus power results and the EMS sets it to be delivered to the dump load, as observed by the $P_d(t)$ response in Fig. \ref{fig:case1_P} (b). Next, during $t=6~\mathrm{s}$ to $t=8~\mathrm{s}$ the irradiance decreases from $1000~\mathrm{W/m^2}$ to $300~\mathrm{W/m^2}$, and as a result, $P_{pv}(t)$ decreases to $29.5~\mathrm{kW}$. The EMS once again sets $P_{fc}^{*}(t)$ to $P_{fc}^{rated}$, bringing the total generation of the hybrid system to $129.5~\mathrm{kW}$ and an unmet deficit of $20.5~\mathrm{kW}$ results. In the next interval i.e., $t=8~\mathrm{s}$ to $t=10~\mathrm{s}$, $P_{pv}(t)$ goes back up to $100~\mathrm{kW}$ when the irradiance is restored to $1000~\mathrm{W/m^2}$, and the deficit of $50~\mathrm{kW}$ is thoroughly met by the FC generator.

It is seen that both the current and voltage controllers track their references with a fast transition and respond to changes very quickly. $Q_{grid}(t)$ is tightly regulated around its $0~\mathrm{kVAR}$ reference, as seen in Fig. \ref{fig:case1_QVdc} (a). The transients that occur during power transitions are reflected in $v_{dc}(t)$ response as seen in Fig. \ref{fig:case1_QVdc} (b), but the DRC controller quickly regulates $v_{dc}(t)$ back to its $800~\mathrm{V}$ reference, demonstrating its disturbance rejection traits. 

\subsection{Case 2: PQ Support} \label{ssec:pq}
In this test case, $P^{*}(t)$ profile remains the same as in case 1, but now the GO demands reactive power support from the hybrid PV-FC system according to the following profile: $100~\mathrm{kVAR}$ from $t=0~\mathrm{s}$ to $t=2~\mathrm{s}$ and from $t=6~\mathrm{s}$ to $t=10~\mathrm{s}$, and $200~\mathrm{kVAR}$ from $t=2~\mathrm{s}$ to $t=6~\mathrm{s}$. $P_{grid}(t)$, $P_{pv}(t)$, $P_{pv}(t)$, $P_{fc}(t)$ and $P_{d}(t)$ are plotted in Fig. \ref{fig:case2_P} while $Q_{grid}(t)$, $Q^{*}(t)$, $v_{dc}(t)$ and $v_{dc}^{*}(t)$ are plotted in Fig. \ref{fig:case2_QVdc}. 
\begin{figure}[t!]
	\centering
	\includegraphics[trim={0 0 0 0},scale=0.65]{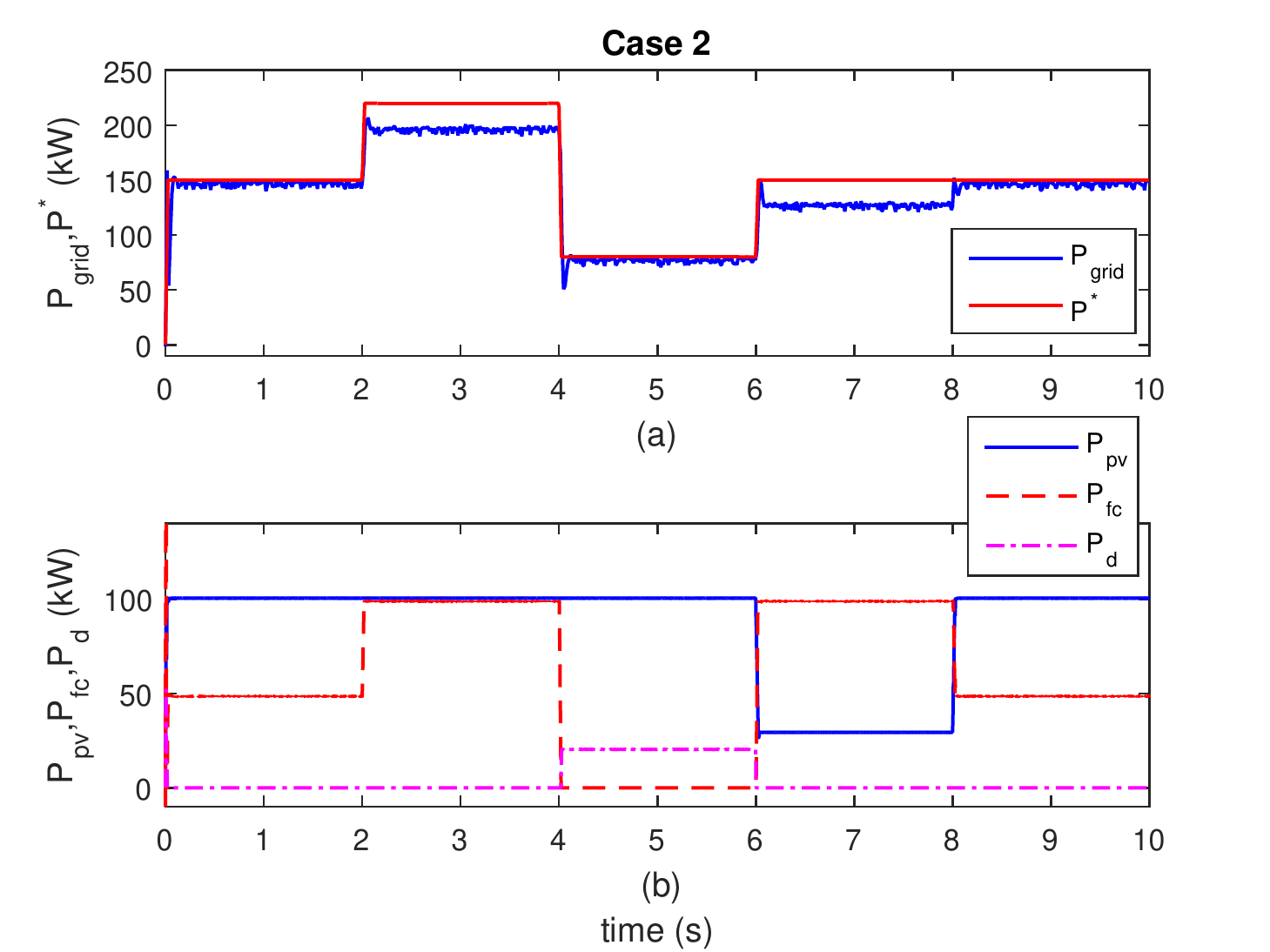}
	\caption{Case 2: PQ support: (a) $P_{grid}(t)$ and $P^*(t)$, and (b) $P_{pv}(t)$, $P_{fc}(t)$ and $P_{d}(t)$.}
	\label{fig:case2_P}
\end{figure}
\begin{figure}[t!]
	\centering
	\includegraphics[trim={0 0 0 0},scale=0.65]{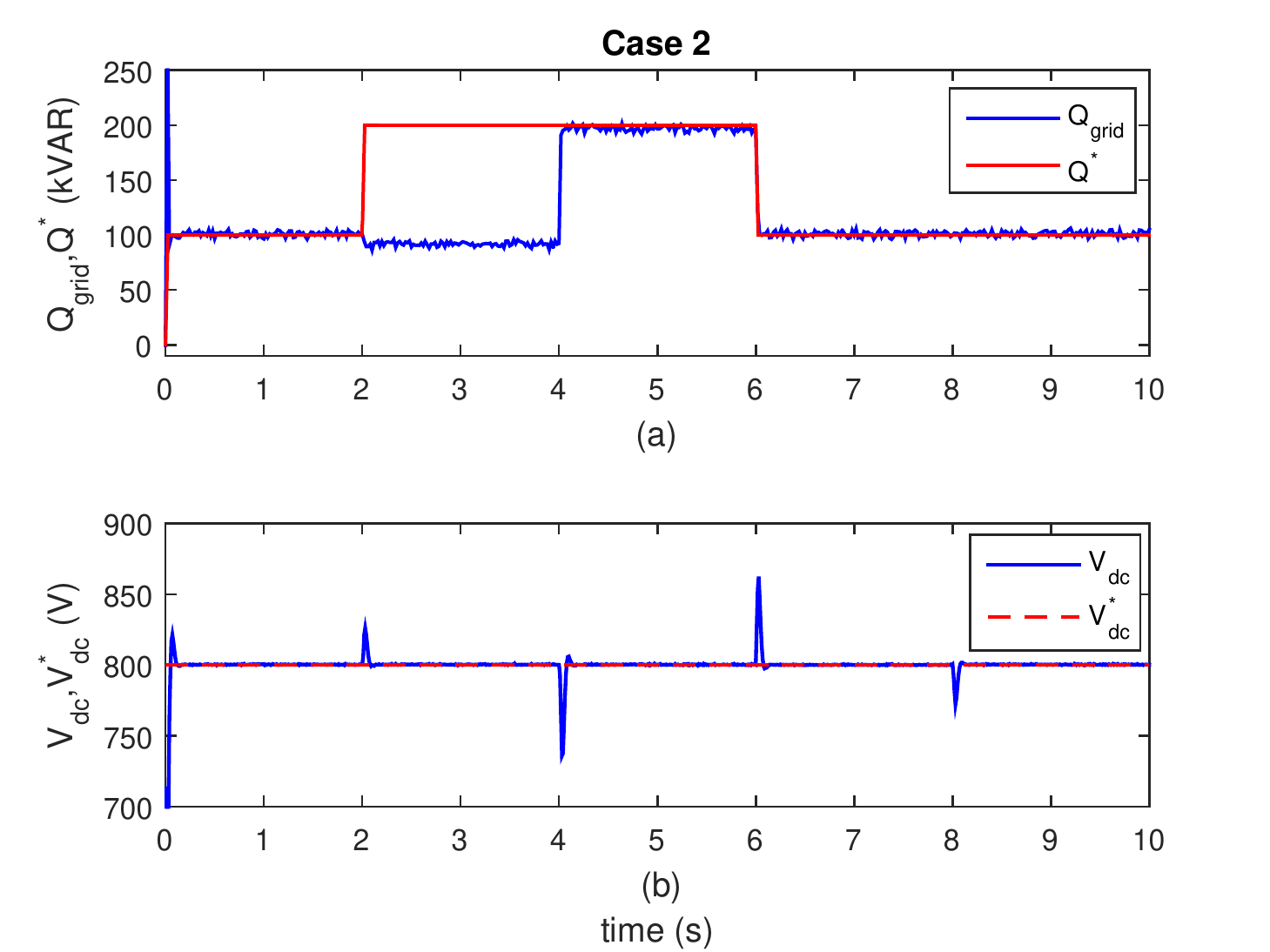}
	\caption{Case 2: PQ support: (a) $Q_{grid}(t)$ and $Q^*(t)$, and (b) $v_{dc}(t)$ and $v_{dc}^{*}(t)$.}
	\label{fig:case2_QVdc}
\end{figure} 

For this test case, $P_{grid}(t)$ has the same profile as in case 1, as seen in Fig. \ref{fig:case2_P}. During the first interval i.e., from $t=0~\mathrm{s}$ to $t=2~\mathrm{s}$, the reactive power demand $Q^{*}(t)$ is fully met by the hybrid system as shown by $Q_{grid}(t)$ response in Fig. \ref{fig:case2_QVdc} (a). Next, when $Q^{*}(t)$ is stepped up at $t=2~\mathrm{s}$, $Q_{grid}(t)$ decreases to $92.5~\mathrm{kVAR}$ since $P^{*}(t)$ also increases and the EMS adjusts $Q_{grid}^{*}(t)$ to stay within the reactive power limit $S^{max}$ of $220~\mathrm{kVA}$, as shown in the flowchart of Fig. \ref{fig:flowchart}. An unmet reactive power deficit of $57.5~\mathrm{kVAR}$ results due to the fact that the EMS prioritizes delivery of $P_{grid}(t)$ over $Q_{grid}(t)$. The deficit is met in the next interval, from $t=4~\mathrm{s}$ to $t=6~\mathrm{s}$ when $P^{*}(t)$ reduces to $80~\mathrm{kW}$ and the EMS deems it safe to deliver the prescribed amount ($150~\mathrm{kVAR}$) of reactive power to the grid. For the remainder of the test case i.e., from $t=6~\mathrm{s}$ to $t=10~\mathrm{s}$, $Q^{*}(t)$ stays at $100~\mathrm{kVAR}$, and is safely delivered by the hybrid system to the grid. Voltage and current controllers exhibit fast responses, and efficient tracking and disturbance rejection traits as seen in both Figs. \ref{fig:case2_P} and \ref{fig:case2_QVdc}. 

\subsection{Case 3: Voltage Sags under Nominal Irradiance} \label{ssec:faults1}
For this test case, three different voltage sags under nominal irradiance of  $1000~\mathrm{W/m^2}$ are applied at regular intervals: single-phase-to-ground (1PG) voltage sag from $t=1~\mathrm{s}$ to $t=3~\mathrm{s}$ with $30\%$ drop in phase $a$ voltage, two-phase-to-ground (2PG) voltage sag from $t=4~\mathrm{s}$ to $t=6~\mathrm{s}$ with $35\%$ drop in phase $a$ and $b$ voltages, and three-phase-to-ground (3PG) voltage sag from $t=7~\mathrm{s}$ to $t=9~\mathrm{s}$ with $40\%$ voltage drop in all the three phases. $P^{*}(t)$ and $Q^{*}(t)$ are set to a constant $150~\mathrm{kW}$ and $0~\mathrm{kVAR}$ respectively. During voltage sags, references $P_{grid}^{*}(t)$ and $Q_{grid}^{*}(t)$ are calculated by the EMS considering sag magnitude, VSC current limits and grid support requirements. $P_{grid}(t)$ and $Q_{grid}(t)$ are plotted in Fig. \ref{fig:case3_PQ} while $P_{pv}(t)$, $P_{fc}(t)$, $P_{d}(t)$, and $v_{dc}(t)$ are plotted in Fig. \ref{fig:case3_PVdc}. Grid voltages and currents during the asymmetrical 1PG and 2PG voltage sags are plotted in Fig. \ref{fig:case3_VI}.
\begin{figure}[t!]
	\centering
	\includegraphics[trim={0 0 0 0},scale=0.65]{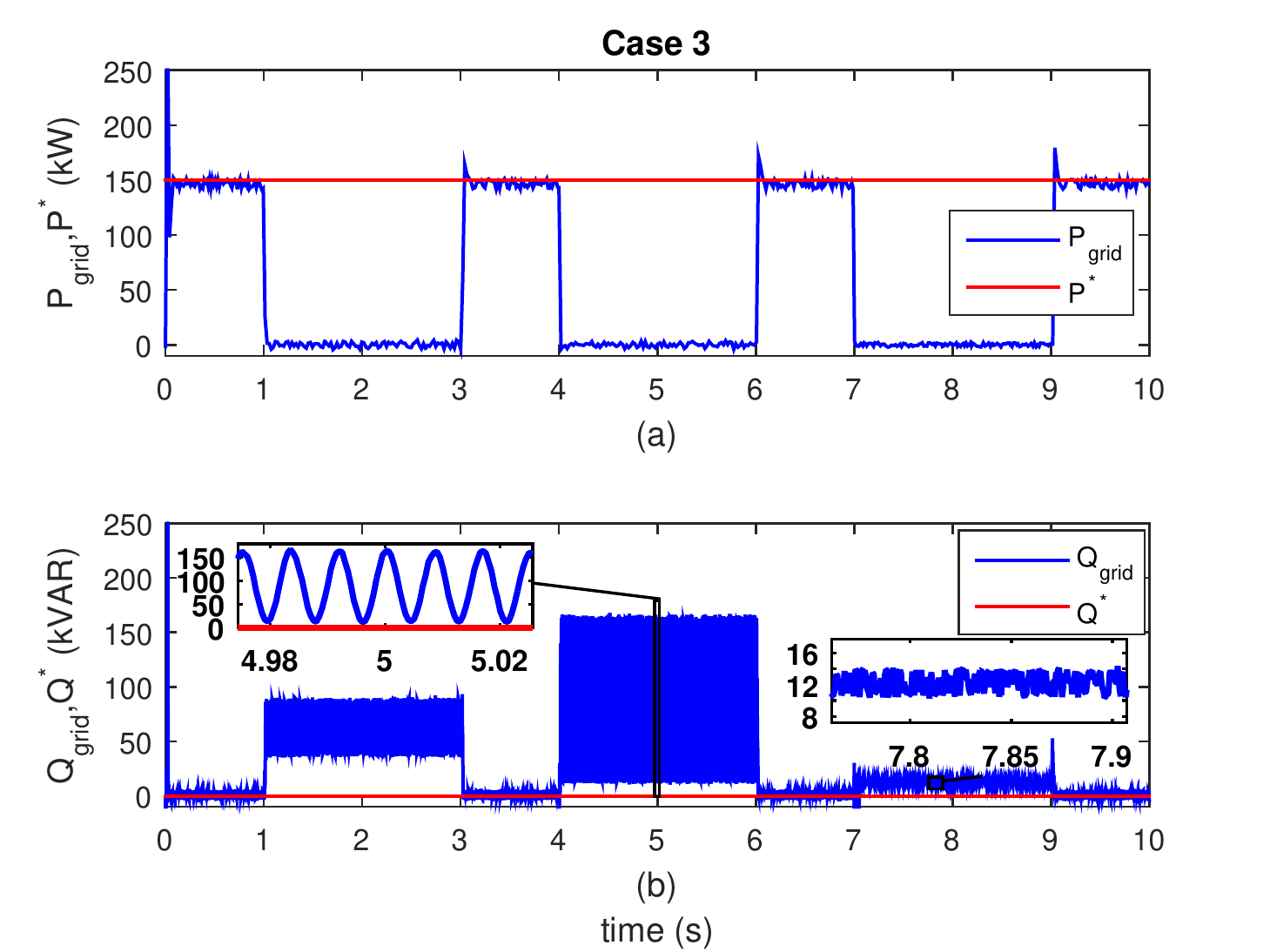}
	\caption{Case 3: Voltage sags @ $1000~\mathrm{W/m^2}$: (a) $P_{grid}(t)$ and (b) $Q_{grid}(t)$.}
	\label{fig:case3_PQ}
\end{figure}
\begin{figure}[t!]
	\centering
	\includegraphics[trim={0 0 0 0},scale=0.65]{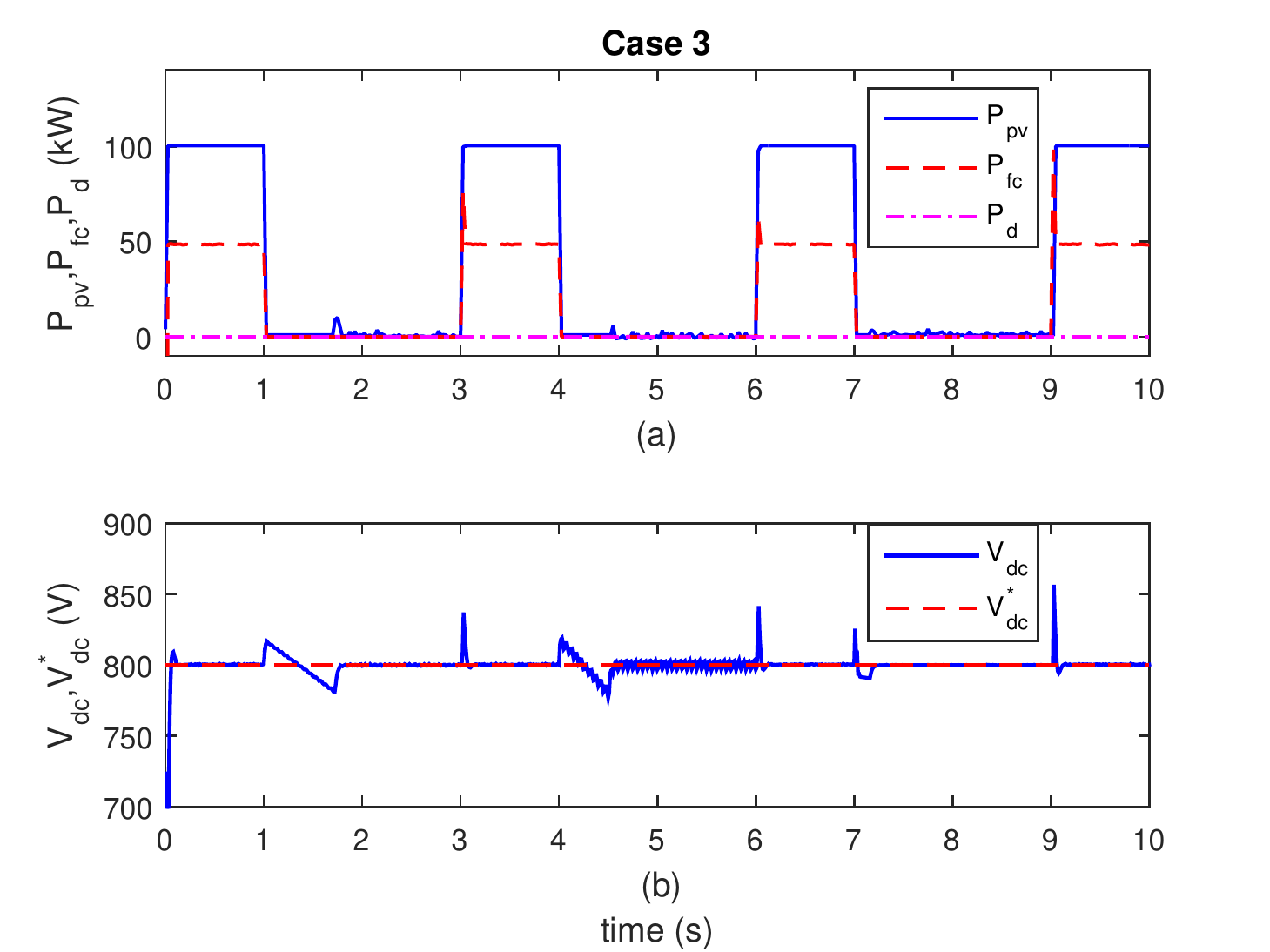}
	\caption{Case 3: Voltage sags @ $1000~\mathrm{W/m^2}$: (a) $P_{pv}(t)$, $P_{fc}(t)$ and $P_{d}(t)$, and (b) $v_{dc}(t)$ and $v_{dc}^{*}(t)$.}
	\label{fig:case3_PVdc}
\end{figure} 
\begin{figure}[t!]
	\centering
	\includegraphics[trim={0 0 0 0},scale=0.65]{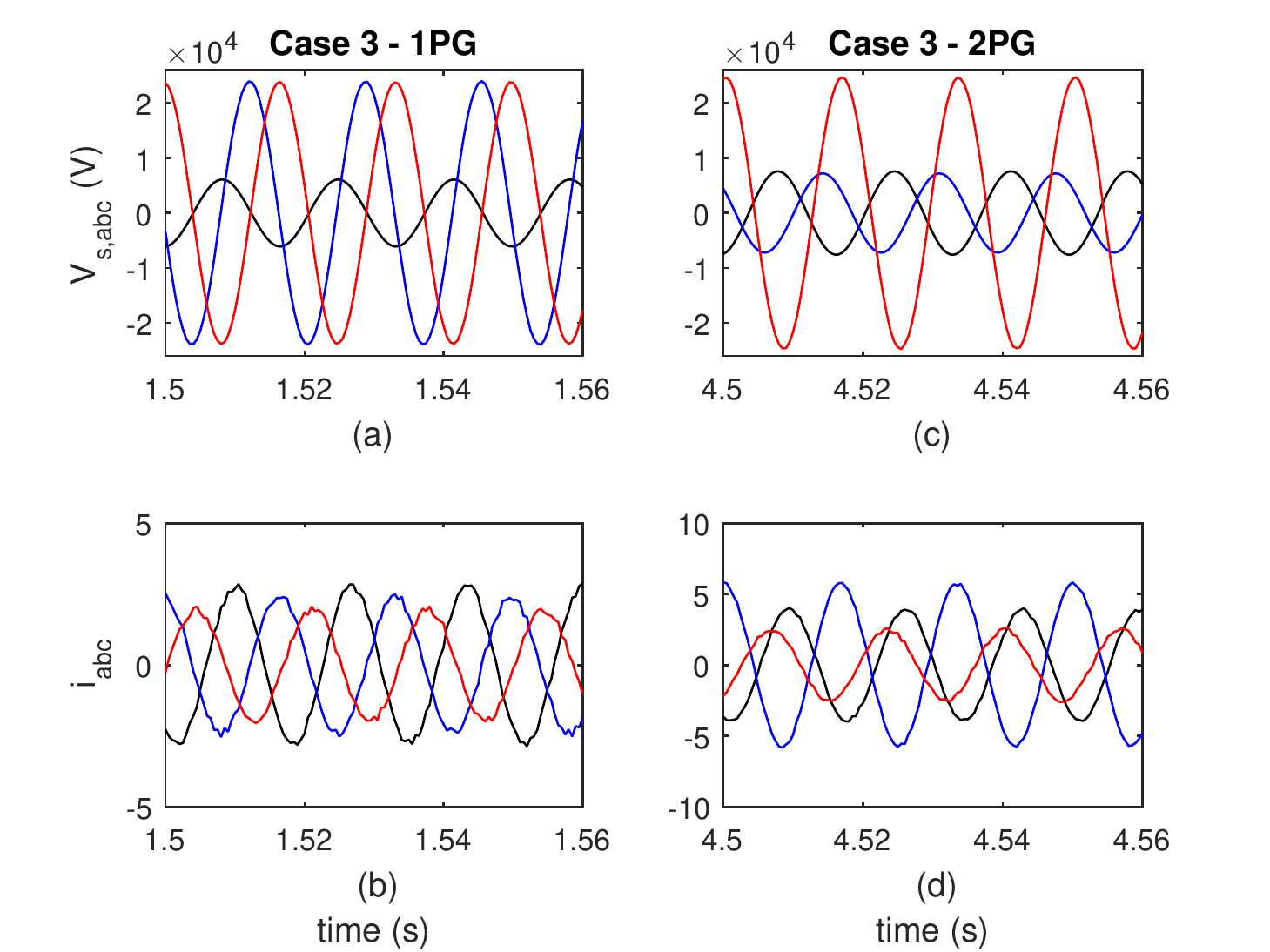}
	\caption{Case 3: Voltage sags @ $1000~\mathrm{W/m^2}$: (a) grid voltages (1PG), (b) grid currents (1PG), (c) grid voltages (2PG), and (d) grid currents (2PG).}
	\label{fig:case3_VI}
\end{figure} 

As seen in Fig. \ref{fig:case3_PQ} (a), the hybrid system delivers the prescribed $150~\mathrm{kW}$ and $0~\mathrm{kVAR}$ during normal operation and PV operates at MPPT. During the voltage sag intervals i.e., from $t=1~\mathrm{s}$ to $t=3~\mathrm{s}$, $t=6~\mathrm{s}$ to $t=4~\mathrm{s}$ and $t=7~\mathrm{s}$ to $t=9~\mathrm{s}$, $P_{grid}^{*}(t)$ is curtailed by the EMS to keep the injected currents within allowable limits. At the same time, it increases $Q_{grid}^{*}(t)$ to provide dynamic grid support. During the asymmetrical voltage sags, an oscillatory reactive power is injected into the grid with an average value of $75~\mathrm{kVAR}$ for the 1PG and $100~\mathrm{kVAR}$ for the 2PG sag, while during the symmetrical 3PG voltage sag, a constant $Q_{grid}(t)$ of $14~\mathrm{kVAR}$ is injected for grid support, as seen in Fig. \ref{fig:case3_PQ} (b).  

Figure \ref{fig:case3_PVdc} (a) indicates that the EMS effectively coordinates the real power generation between PV and FC generators during both normal and sag conditions. $P_{pv}(t)$ supplies for the bulk of $P^{*}(t)$ with the deficit coming from $P_{fc}(t)$ under normal conditions. Power generation of both the PV and FC generators is curtailed; PV power generation is reduced to zero by the DRC controlling its boost converter, while the FC power generation is curtailed to zero as its controller tries to follows the $P_{fc}^{*}(t)=0~\mathrm{kW}$ command from the EMS. The response of $v_{dc}(t)$ in Fig. \ref{fig:case3_PVdc} (b) shows transient overshoots and undershoots at sag occurrence instances as the controller tries to maintain active power balance between dc-link and the grid. The voltage is quickly recovered and regulated to its reference once the power balance is restored. During the 1PG and 2PG faults, some oscillations appear in $v_{dc}(t)$ in the intervals from $t=1~\mathrm{s}$ to $t=3~\mathrm{s}$ and from $t=4~\mathrm{s}$ to $t=6~\mathrm{s}$ due to the presence of negative sequence components in grid voltages and currents. The DRC successfully suppresses those oscillations for the 1PG fault and restricts them to within $\pm 0.5\%$ of $v_{dc}^{*}(t)$ during the 2PG fault. It is also worth noting from Fig. \ref{fig:case3_PQ} (a) and Fig. \ref{fig:case3_VI} (b) and (d) that the real power remains constant and the grid currents remain sinusoidal with a low harmonic content during the 1PG and 2PG voltage sags, owing to the time delay reference generation \cite{Liu2016a}, and sinusoidal current tracking and harmonic suppression through RC without having to use a PLL. The total harmonic distortion (THD) in grid currents is $4.41\%$ during the 1PG and $3.02\%$ during the 2PG voltage sag.

\subsection{Case 4: Voltage Sags under Low Irradiance} \label{ssec:faults2}
For this test case, 1PG, 2PG and 3PG voltage sags are applied under a low irradiance of $300~\mathrm{W/m^2}$ during the same intervals as in case 3. $P^{*}(t)$ and $Q^{*}(t)$ are set to a constant $150~\mathrm{kW}$ and $0~\mathrm{kVAR}$ respectively. $P_{grid}(t)$ and $Q_{grid}(t)$ are plotted in Fig. \ref{fig:case4_PQ} while $P_{pv}(t)$, $P_{fc}(t)$, $P_{d}(t)$, and $v_{dc}(t)$ are plotted in Fig. \ref{fig:case4_PVdc}. Grid voltages and currents during the asymmetrical 1PG and 2PG voltage sags are plotted in Fig. \ref{fig:case4_VI}.
\begin{figure}[t!]
	\centering
	\includegraphics[trim={0 0 0 0},scale=0.65]{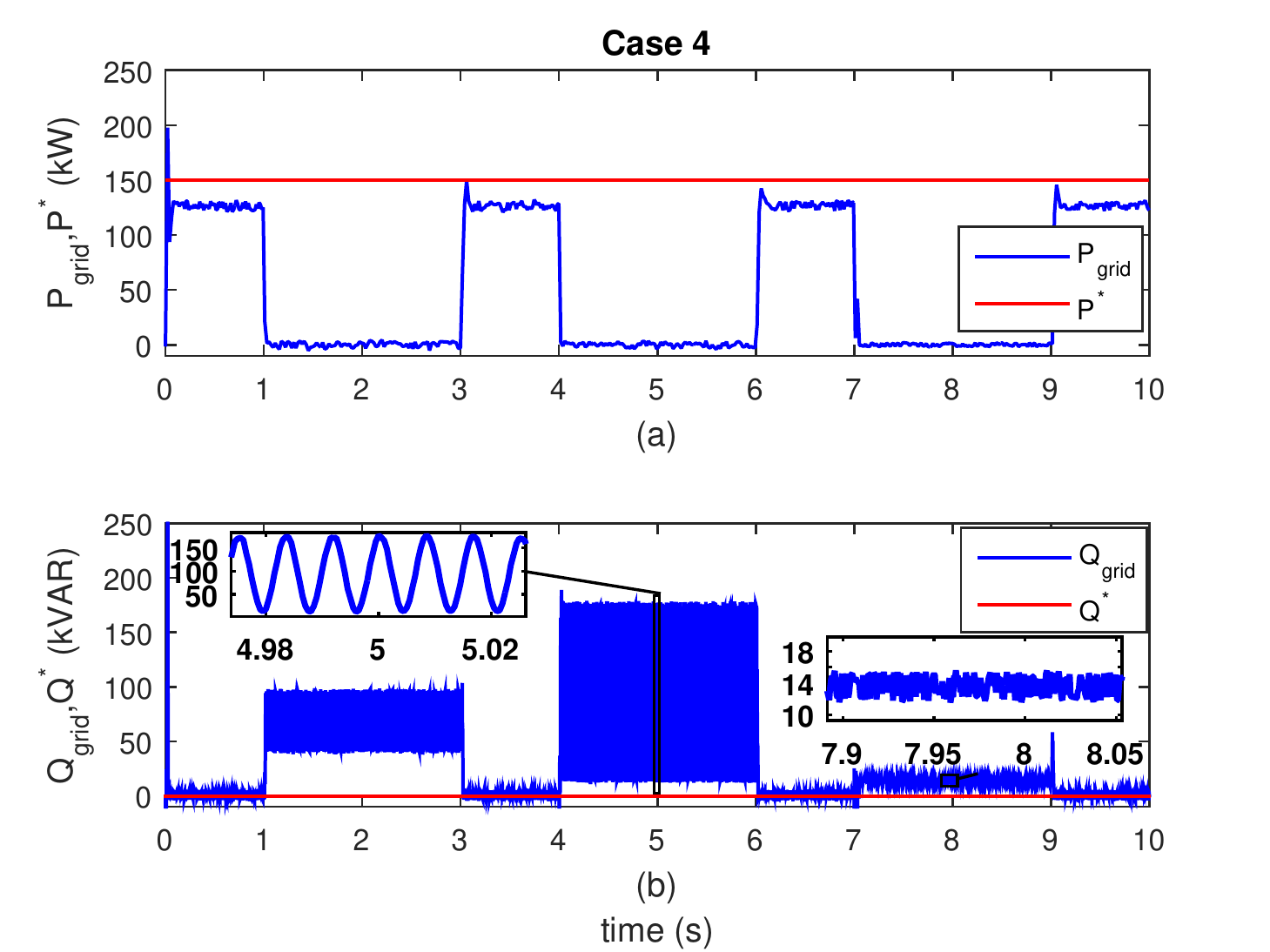}
	\caption{Case 4: Voltage sags @ $300~\mathrm{W/m^2}$: (a) $P_{grid}(t)$ and (b) $Q_{grid}(t)$.}
	\label{fig:case4_PQ}
\end{figure}
\begin{figure}[t!]
	\centering
	\includegraphics[trim={0 0 0 0},scale=0.65]{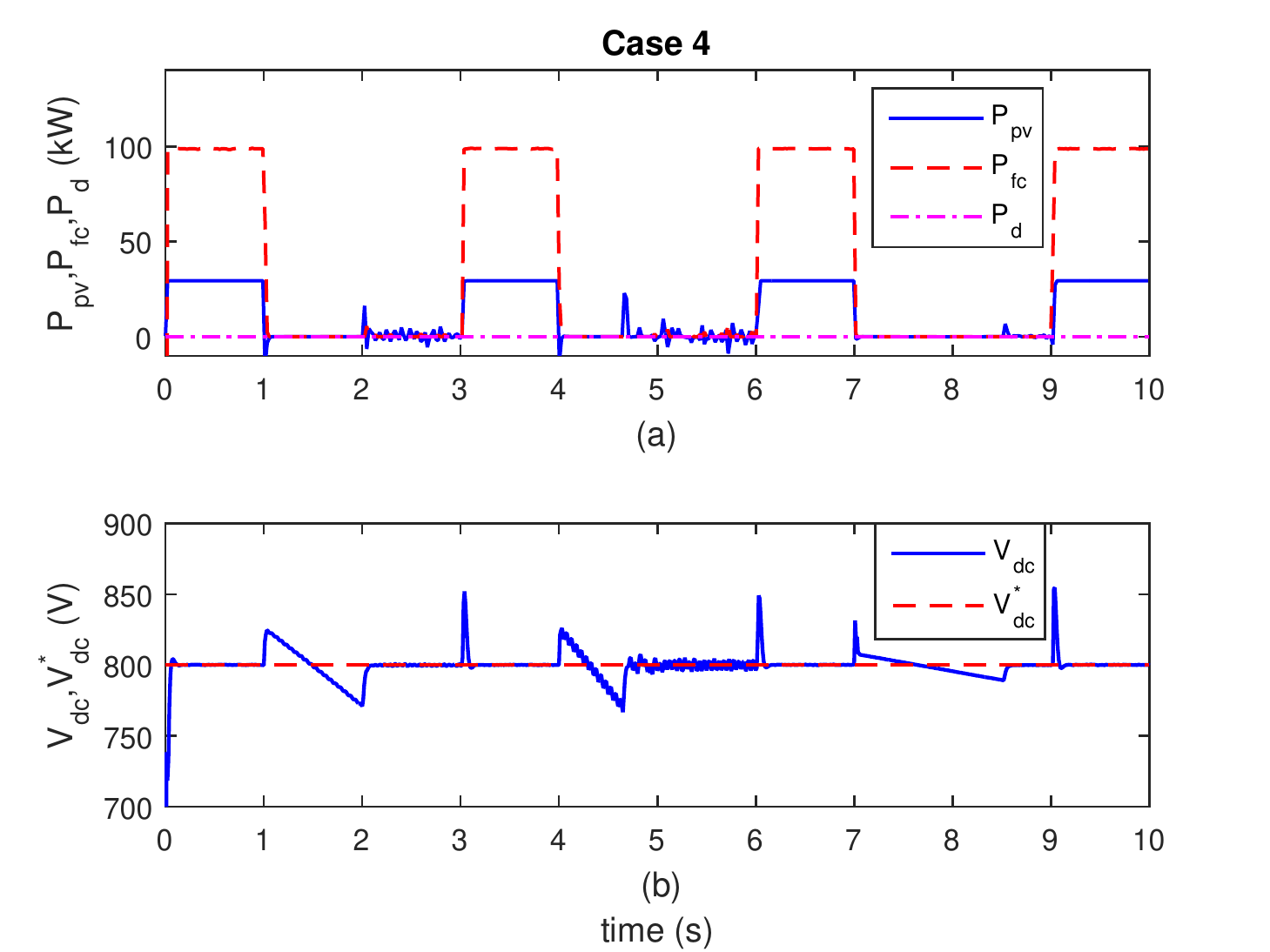}
	\caption{Case 4: Voltage sags @ $300~\mathrm{W/m^2}$: (a) $P_{pv}(t)$, $P_{fc}(t)$ and $P_{d}(t)$, and (b) $v_{dc}(t)$ and $v_{dc}^{*}(t)$.}
	\label{fig:case4_PVdc}
\end{figure}
\begin{figure}[t!]
	\centering
	\includegraphics[trim={0 0 0 0},scale=0.65]{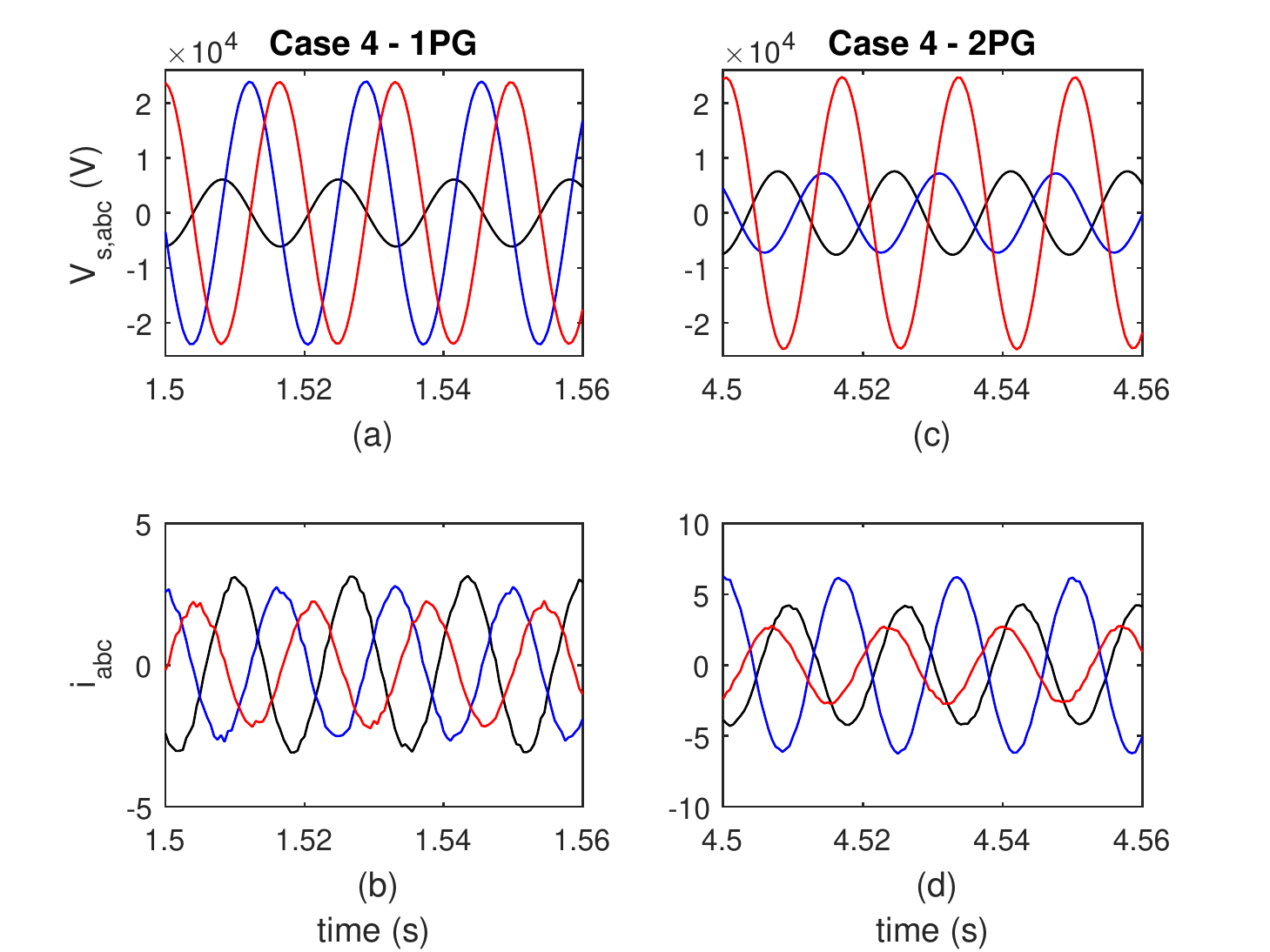}
	\caption{Case 4: Voltage sags @ $300~\mathrm{W/m^2}$: (a) grid voltages (1PG), (b) grid currents (1PG), (c) grid voltages (2PG), and (d) grid currents (2PG).}
	\label{fig:case4_VI}
\end{figure} 

Real power delivery performance of the hybrid PV-FC systems remains similar to case 3, except for a $20~\mathrm{kW}$ deficit during normal operating conditions due to low irradiance of $300~\mathrm{W/m^2}$, as seen in Fig. \ref{fig:case4_PQ} (a). The system still provides reactive power support to the grid during voltage sags in the same amounts as in case 3, shown in Fig. \ref{fig:case4_PQ} (b). The controllers track their references swiftly and efficiently despite the presence of uncertainties and external disturbances, demonstrating the high performance and robustness features of the control system.

Responses of the above test cases show that under the proposed EMS and robust control strategy, the hybrid PV-FC system's exhibits a fast, efficient and robust dynamic response and is successfully able to ride through voltage sags while providing support to the grid. The controllers have low complexity and a carefully tuned PLL is not required to extract PN sequence components during asymmetrical voltage sags, thus, reducing the design and computational requirements without sacrificing performance. 

\section{Conclusion} \label{sec:conclusion}
An efficient EMS with a robust control strategy is presented for a grid-connected hybrid PV-FC power system. Under the proposed scheme, the hybrid power plant is able to effectively operate under both normal conditions and voltage sags. The scheme enables the hybrid system to provide dynamic support to the grid during abnormal conditions, thus making it compatible with the modern grid codes and suitable for adoption in an environment having high penetration of distributed generation. The control system is systematically designed by applying Lyapunov theory and solving LMI constraints, making it convenient and easy to synthesize. Robustness features are incorporated into the control design not only through a formal mathematical treatment but by designing the controllers independent of the dynamics of PV or FC models. The control laws have low complexity and don't require a PLL, that imposes additional design and computational requirements on the system. Performance of the EMS and its control system are demonstrated through a number of test cases to demonstrate their efficacy, and robustness disturbance rejection features.  

\section*{References}

\bibliography{library}


\end{document}